\documentclass[11pt]{amsart}

\usepackage{amssymb,amsmath,mathtools} 
\usepackage{graphicx} 

\usepackage[hmargin=0.12\paperwidth,vmargin=0.2\paperwidth,bindingoffset=0.5cm]{geometry}

\newtheorem{thm}{Theorem}[section]
  \newtheorem*{thm*}{Theorem}

\newtheorem{cor}[thm]{Corollary}
\newtheorem{lem}[thm]{Lemma}
\newtheorem{prop}[thm]{Proposition}
  \newtheorem*{prob*}{Problem}
\theoremstyle{definition}
  \newtheorem{defn}[thm]{Definition}
  \newtheorem*{defn*}{Definition}

\newtheorem{rem}[thm]{Remark}
  
  \newtheorem*{rem*}{Remark}
\numberwithin{equation}{section}

\newcommand{\USp}{\operatorname{USp}}

\newcommand{\C}{\mathbb C}
\newcommand{\R}{\mathbb R}

\newcommand{\q}{\textbf}

\DeclareMathOperator{\per}{per}
\DeclareMathOperator{\quaternion}{qu}

\DeclareMathOperator{\sgn}{sgn}
\DeclareMathOperator{\GAMMA}{Gamma}

\DeclareMathOperator{\Prob}{Pr}

\DeclareMathOperator{\Pf}{Pf}

\DeclareMathOperator{\Poisson}{Pois}
\DeclareMathOperator{\diag}{diag}
\DeclareMathOperator{\Mat}{Mat}
\DeclareMathOperator{\Hq}{\mathbb{H}}

\newcommand{\Tr}{\mathop{\mathrm{Tr}}}

\newcommand{\e}{\mbox{e}}

\usepackage{hyperref}   
\hypersetup{
pdfborder={0 0 0},      
}

\begin{document}
\title[Permanental processes from products of random matrices]
 {\bf{ Permanental processes from products of complex and quaternionic induced Ginibre ensembles}}

\author{Gernot Akemann}
  \address{Department of Physics,  Bielefeld University,  P.O. Box 100131, D-33501 Bielefeld, Germany}
  \email{akemann@physik.uni-bielefeld.de}
\author{Jesper R. Ipsen}
  \address{Department of Physics, Bielefeld University,  P.O. Box 100131, D-33501 Bielefeld, Germany}
  \email{jipsen@physik.uni-bielefeld.de}
\author{Eugene Strahov}
  \address{Department of Mathematics, The Hebrew University of Jerusalem, Givat Ram, Jerusalem 91904, Israel}
  \email{strahov@math.huji.ac.il}


\commby{}


\begin{abstract}

We consider products of independent random matrices taken from the induced Ginibre ensemble with complex or quaternion elements. The joint densities for the complex eigenvalues of the product matrix can be written down exactly for a product of any fixed number of matrices and any finite matrix size. We show that the squared absolute values of the eigenvalues form a permanental process, generalising the results of  Kostlan and Rider for single matrices to products of complex and quaternionic matrices. Based on these findings, we can first write down exact results and asymptotic expansions for the so-called hole probabilities, that a disc centered at the origin is void of eigenvalues. Second, we compute the asymptotic expansion for the opposite problem, that a large fraction of complex eigenvalues occupies a disc of fixed radius centered at the origin; this is known as the overcrowding problem. While the expressions for finite matrix size depend on the parameters of the induced ensembles, the asymptotic results agree to leading order 
with previous results for products of square Ginibre matrices.

\end{abstract}

\maketitle

\raggedbottom

\section{Introduction}
In random matrix theory, the classical Ginibre ensembles given by real ($\beta=1$), complex ($\beta=2$) and quaternion ($\beta=4$) matrix elements with Gaussian distribution without further symmetry
play a central role and enjoy many applications. For recent reviews see~\cite{KhoruzhenkoSommers,FS} and references therein. Such ensembles lead to determinantal and Pfaffian point processes on the complex plane, forming a true Coulomb gas in two dimensions, see e.g. \cite{Forrester0} for this interpretation.

Recently there has been a considerable interest in the distribution of eigenvalues of \textit{products} of a finite number of independent random matrices of finite size $N$, taken from the classical Ginibre ensembles.
While the density in the limit $N\to\infty$ was previously known, see \cite{Burda,goetze}, an exact solution for finite $N$ was lacking. For applications of such products of random matrices we refer to \cite{Burdarev} and references within.
In a series of papers the authors and their co-workers
showed that eigenvalues of products of matrices taken from the complex \cite{ABu} and quaternionic \cite{Ipsen}
Ginibre ensemble form determinantal or Pfaffian point processes, respectively.
In \cite{AkemannStrahov}  the results for complex Ginibre matrices were used to compute hole probabilities and overcrowding probabilities for this determinantal point process and its infinite analogue.
Products of real Ginibre matrices were studied by  Forrester \cite{Forrester} who calculated the probability that all eigenvalues of such a product are real.
Products of matrices from other ensembles or from different ensembles have been considered too, including inverse complex Ginibre matrices  by Adhikari,  Reddy,  Reddy, and Saha \cite{AdhikariReddyReddySaha}, and truncated unitary matrices in \cite{ABKN} (the ensemble of truncated unitary matrices was introduced by \.Zyczkowski and Sommers \cite{KS}). Rectangular matrices  \cite{AdhikariReddyReddySaha} as well as most recently products of elliptic Ginibre matrices \cite{RRSV} have been considered, and the order in which different (square) matrices are multiplied is not important in a weak sense, even at finite $N$
\cite{IpsenKieburg}.
For mixing Hermitian and non-Hermitian matrices in the macroscopic limit and the corresponding limiting densities see \cite{Burdarev}.
Furthermore, the exact results for finite $N$ allowed to study correlation functions on a local scale. Whereas in the bulk and at the edge these were found to be in the corresponding Ginibre universality class \cite{ABu,ABKN} or weak non-unitary class \cite{ABKN}, at the origin new universality classes labelled by the number of matrices multiplied were found \cite{ABu,ABKN}.

Although we will focus on the complex eigenvalues, we mention in passing that the distribution of singular values of the product matrix and all correlation functions have been calculated in~\cite{AkemannKieburgWei,AIK}. In particular, it was shown that the squared singular values form a biorthogonal ensemble in the sense of Borodin \cite{BorodinBiorthogonal}, being distributed according to a determinant point process.
The limiting correlation kernel at the hard edge of this determinantal point process was studied by Kuijlaars and Zhang \cite{Zhang,KuijlaarsZhang}. Their integral representation of this kernel for the product of two matrices was found to coincide with that of a Cauchy--Laguerre two-matrix model~\cite{BGS:2009,BGS:2014}. Very recently, differential equations for gap probabilities for this type of biorthogonal ensembles have been derived~\cite{Strahov:2014}. 

The present paper extends some of the aforementioned results to products of independent random matrices taken from \textit{induced} Ginibre ensembles with complex or quaternion elements.
The induced Ginibre ensemble of a single random matrix introduced by  Fischmann, Bruzda, Khoruzhenko, Sommers, and $\dot{\mbox{Z}}$yczkowski \cite{Fischmann} for $\beta=1,2$ simply results from the corresponding classical Ginibre ensembles when considering rectangular matrices. Their aim was to describe statistical properties of evolution operators in quantum mechanical systems.
In the complex case $\beta=2$ the elliptic Ginibre ensemble with inserted determinants was previously considered and solved in \cite{A01}, generalising the corresponding induced Ginibre ensemble. There the focus was on applications to quantum chromodynamics with chemical potential in three dimensions.
Furthermore, motivated again from quantum chromodynamics the product of two random matrices from the induced ensemble was considered for $\beta=2$ and $4$ in \cite{APS}. There the exact hole probabilities and their asymptotic expansions were calculated, which allows us to compare with these special cases.
The induced ensembles are of mathematical interest for the following reasons. First, these ensembles are special cases of
the non-Hermitian Feinberg--Zee ensembles with a specific choice of potential that can be exactly solved for finite $N$. Second, the induced Ginibre ensembles provide an explicit realisation of the single ring theorem, namely when the large $N$ limit is taken such that the difference between the long and short side of the rectangular matrix is of the order of $N$ \cite{Fischmann}.
For a discussion of such ensembles with arbitrary potentials and the single ring theorem we refer the reader to Feinberg and Zee \cite{Feinberg0}, Feinberg, Scalettar and  Zee \cite{Feinberg1},
Feinberg \cite{Feinberg2}, Guionnet, Krishnapur, and Zeitouni \cite{Guionnet} and references therein.

The paper is organized as follows. In Section \ref{SectionComplexInducedGinibreResults} we state all results for the product of $n$ independent random matrices taken from the complex induced Ginibre ensemble at $\beta=2$. Our starting point is the joint density of eigenvalues from \cite{ABu,IpsenKieburg} given by a determinantal point process on the complex plane with the weight being a Meijer $G$-function in Theorem \ref{PropositionJointDensityOfEigenvalues}. This result was rigorously established in \cite{AdhikariReddyReddySaha}, and its kernel easily follows. We show that the set of squared absolute values of the complex eigenvalues forms a permanental process, thus generalising the work of Kostlan \cite{Kostlan} for $n=1$, and of \cite{AkemannStrahov}.
This enables us to compute the hole probability that a disc of radius $r$ centered at the origin is empty for $N$ finite. Both for finite and infinite $N$ we then take the asymptotic limit $r\to\infty$. For $N\to\infty$, we establish bounds and compute the leading order asymptotic for the probability that a large number $q$ of eigenvalues lie in a centered disk of fixed radius $r$ known as overcrowding. In Section \ref{SectionResultsQuaternion} we state the corresponding results for quaternionic matrix elements at $\beta=4$. We first give a rigorous proof for the joint density established in \cite{Ipsen}.
It uses the technique of differential forms similar to that presented by Hough, Krishnapur, Peres and Vir\'ag \cite{Hough} in the context of the classical complex Ginibre ensemble of single matrices. Also for $\beta=4$, we find a permanental process for the moduli generalising the work of Rider~\cite{Rider}. The main new result in this section is Theorem~\ref{TheoremQ|z_1||z2|Distribution}, which says that the set of moduli of eigenvalues has the same distribution as a set of independent random  variables. Furthermore, these random variables can be described in terms of products of independent gamma variables.
The corresponding expressions for the hole probabilities and overcrowding estimates at $\beta=4$ can be  obtained from Theorem \ref{TheoremQ|z_1||z2|Distribution}.
The rest of the paper is devoted
to proofs: The proofs for the theorems presented in Section~\ref{SectionComplexInducedGinibreResults} are given in Sections \ref{SectionStartOfProofs} to \ref{SectionProofsOvercrowding}, while the proofs for the theorems presented in Section~\ref{SectionResultsQuaternion} are given in Sections \ref{sec:quaternions} to \ref{sec:holes-beta4}. In Appendix \ref{App} we recall the higher order terms in the Euler--Maclaurin summation  formula that we need in the main text.


\section{Products of random matrices from the induced complex Ginibre ensemble}
\label{SectionComplexInducedGinibreResults}

\subsection{Definition of the product matrix}

Let $\Mat(N,\C)$ be the set of all $N\times N$ matrices $M$ with entries in $\C$.
Define a probability measure $\mathcal{P}_{N,m}$ on $\Mat(N,\C)$ by the formula
\begin{equation}\label{inducedGin}
\mathcal{P}_{N,m}(dM)\equiv\frac{1}{Z_{N,m}}\left[\det(M^{*}M)\right]^m\e^{-\Tr (M^{*}M)}dM,
\end{equation}
where $m\geq 0$ is a parameter, $Z_{N,m}$ is a normalisation constant, and $dM$ is the Lebesgue measure on  $\Mat(N,\C)$.
This parameter dependent probability measure $\mathcal{P}_{N,m}$ is called the induced complex Ginibre ensemble  (with parameter $m$) and  was studied in detail in \cite{Fischmann}.
Note that $\mathcal{P}_{N,m}$ is a special case of the non-Hermitian Feinberg--Zee ensemble \cite{Feinberg0} corresponding to the potential $V(t)=t+m\log t$ with $t=M^* M$.
If $m=0$, then this ensemble is reduced to the classical Ginibre ensemble of complex random matrices \cite{Ginibre}.

Let $n$ be a positive integer, $m_1,\ldots, m_n$ be positive real numbers, and consider the product $P_n$ of $n$ independent
matrices $M_1, \ldots, M_n$,
\[
P_n=M_1M_2\ldots M_n.
\]
Each matrix $M_a$, $a=1,\ldots,n$, is of size $N\times N$, and it is chosen independently from the induced complex  Ginibre ensemble
with parameter $m_a$. We will refer to $P_n$ as to the product of $n$ random matrices from the induced complex Ginibre ensemble
with the parameters $m_1,\ldots,m_n$. Note that the effect of rectangular matrices can be incorporated into the parameters $m_1,\ldots,m_n$, see~\cite{IpsenKieburg}.  


\subsection{The joint density of eigenvalues as a determinantal point process}

The first result concerns the joint density of eigenvalues of $P_n$. To present this result  we use the same notation and definitions
for the Meijer $G$-function as in Luke \cite{Luke}, Section 5.2.
Namely, the Meijer $G$-function
is defined there as
\begin{equation}
G_{p,q}^{m,n}\left(x\,\biggl|
\begin{array}{cccc}
a_1,  a_2,  \ldots,  a_p \\
b_1,  b_2,  \ldots,  b_q
\end{array}
\right)
\equiv\frac{1}{2\pi i}\int\limits_{C}\frac{\prod_{j=1}^m\Gamma(b_j-s)\prod_{j=1}^n\Gamma(1-a_j+s)}
{\prod_{j=m+1}^q\Gamma(1-b_j+s)\prod_{j=n+1}^p\Gamma(a_j-s)}\ x^sds.
\nonumber
\end{equation}
An empty product is to be interpreted as unity, $0\leq m\leq q$, $0\leq n\leq p$, and the parameters
$\{a_k\}$ ($k=1,\ldots,p$) and $\{b_j\}$ ($j=1,\ldots,m$) are chosen such that no pole of $\Gamma(b_j-s)$ coincides with any of the poles of
$\Gamma(1-a_k+s)$. We also assume that $x\in\C\setminus\{0\}$. The integration contour $C$ goes from $-i\infty$ to $+i\infty$ such that all the poles of
$\Gamma(b_j-s)$, $j=1,\ldots,m,$ lie to the right of the integration path, and that all of the poles of $\Gamma(1-a_k+s)$, $k=1,\ldots,n,$ lie to the left of this path.
When $p=0$, then also $n=0$, and  we can write the corresponding Meijer $G$-function as
$G_{0,q}^{m,0}(x\,|\,b_1,  b_2,  \ldots,  b_q )$. Here the indices $a_p$ are absent and thus omitted.

\begin{thm}
\label{PropositionJointDensityOfEigenvalues}
The vector of (unordered) eigenvalues of $P_n$
has the density (with respect to the Lebesgue measure on $\C^N$)
which can be written as
\begin{align}\label{JointDensityOfEigenvalues}
\varrho_N^{(m_1,\ldots,m_n)}(z_1,\ldots,z_N)
&=\mathcal C
\prod\limits_{k=1}^Nw_n^{(m_1,\ldots,m_n)}(z_k)
\prod\limits_{1\leq i<j\leq N}|z_i-z_j|^2, \\
\mathcal C^{-1}&=N!\pi^{nN}\prod\limits_{k=1}^N \prod\limits_{a=1}^n\Gamma(m_a+k). \nonumber
\end{align}
It is also often called joint probability distribution function. Here, the weight function (to be normalised later) is given by
\begin{equation}\label{WeightFunction}
w_n^{(m_1,\ldots,m_n)}(z)=\pi^{n-1}G_{0,n}^{n,0}(\,|z|^2 \,|\, m_1, m_2, \ldots, m_n ),
\end{equation}
where we have used the Meijer $G$-function with the choice of parameters specified earlier.
\end{thm}
Theorem \ref{PropositionJointDensityOfEigenvalues} was proven in \cite{AdhikariReddyReddySaha} and independently in \cite{Eugene}. For completeness we will derive the form of the weight function in Section \ref{SectionStartOfProofs} as it was not explicitly given in  \cite{AdhikariReddyReddySaha}.
\begin{rem}
1)
Apart from the weight function and normalisation constant this has exactly the same form as the joint density of the Ginibre ensemble \cite{Ginibre} for $\beta=2$, where the complex eigenvalues repel each other through the absolute value squared of a Vandermonde determinant.\newline
2)
In particular when $n=1$ the weight function is given by $w_{n=1}^{(m_1)}(z)=|\zeta|^{2m_1}\e^{-|z|^2}$, and equation
(\ref{JointDensityOfEigenvalues}) reduces to the formula  for the joint density of eigenvalues of a matrix taken from the induced Ginibre ensemble,
see  equation (19) in \cite{Fischmann}. For $m_1=0$ we are of course back to the classical Ginibre ensemble \cite{Ginibre}.\newline
3)
When $n=2$ we have  $w_{n=2}^{(m_1,m_2)}(z)=2|z|^{2(m_1+m_2)}K_{m_1-m_2}(2|z|)$ in terms of the modified Bessel (or MacDonald) function, and the joint density can be found e.g. in \cite{APS}.\newline
4)
For $m_1=\ldots=m_n=0$ we obtain from equation (\ref{JointDensityOfEigenvalues}) the joint density of eigenvalues of a matrix which is a product of $n$
independent complex Ginibre matrices of size $N\times N$. The eigenvalue distributions for such products were studied  in \cite{ABu,AkemannStrahov}. Formulae (\ref{JointDensityOfEigenvalues}) and (\ref{WeightFunction}) for the joint density of eigenvalues
generalise that obtained in \cite{ABu} and were already considered in \cite{AdhikariReddyReddySaha,IpsenKieburg}.
\end{rem}

Recall that a determinantal point process on $\C$ with kernel $K_N(z,\zeta)$ and normalisable weight function $w(z)$ is a point process on $\C$ given by
\[
\varrho_N(z_1,\ldots,z_N)=\frac{1}{Z_N}\prod_{k=1}^N w(z_k) \prod\limits_{1\leq i<j\leq N}|z_i-z_j|^2
\]
whose $\ell$-point correlation functions  with respect to the weight $w(z)$ are given by
\[
\varrho_\ell(z_1,\ldots,z_\ell)\equiv \frac{N!}{(N-\ell)!} \prod_{k=\ell+1}^N \int d^2z_k\ \varrho_N(z_1,\ldots,z_N)
=\frac{1}{Z_N}
\prod_{k=1}^\ell w(z_k)
\det\left[K(z_i,z_j)\right]_{i,j=1}^\ell,
\]
where $Z_N$ is a suitable normalisation constant.
For a discussion of determinantal point processes, their properties and diverse applications we refer the reader to survey papers by
Borodin  \cite{Borodin}, and by  Hough, Krishnapur, Peres, and Vir$\acute{\mbox{a}}$g \cite{Hough1}.
One can apply standard methods of Random Matrix Theory  (see for example, the books by
Anderson,  Guionnet and Zeitouni \cite{Anderson}, Deift \cite{Deift}, Forrester \cite{Forrester0},  and Pastur and Shcherbina \cite{Pastur} and \cite{handbook}) to Theorem \ref{PropositionJointDensityOfEigenvalues}, use that the weight is only depending on the modulus to see that the corresponding orthogonal polynomials are monic powers. The only non-trivial part is to compute their normalisation, which are given in the following

\begin{thm}\label{TheoremDeterminantalCorrelationkernel}
The eigenvalues of $P_n$ form a determinantal point process in the complex plane
with kernel
\begin{equation}\label{Kernel}
K_N^{(m_1,\ldots,m_n)}(z,\zeta)=\frac{1}{\pi^n}
\sum_{k=0}^{N-1}\frac{(z\overline{\zeta})^k}{\prod_{a=1}^n\Gamma(k+1+m_a)}
\end{equation}
with respect to the
normalised
probability measure
\begin{equation}\label{BackgroundProbabilityMeasure}
\frac{G_{0,n}^{n,0}(\,|z|^2\,|\,m_1,  m_2,  \ldots,  m_n )}{\pi\prod_{a=1}^n\Gamma(m_a+1)}d^2z.
\end{equation}
\end{thm}

Note that the kernel $K_N^{(m_1,\ldots,m_n)}(z,\zeta)$ defined by equation (\ref{Kernel}) can be written as
\[
K_N^{(m_1,\ldots,m_n)}(z,\zeta)=\sum\limits_{k=0}^{N-1}\varphi_k(z)\overline{\varphi_k(\zeta)},
\]
where $\left\{\varphi_k, 0\leq k\leq N-1\right\}$
is an orthonormal set in $L^2\left(\C,\mu\right)$, with respect to the corresponding measure defined by equation (\ref{BackgroundProbabilityMeasure}).
The orthonormal polynomials $\left\{\varphi_k, 0\leq k\leq N-1\right\}$ are 
\begin{equation}\label{VarPHIk}
\varphi_k(z)\equiv\big(h_k^{(m_1,\ldots,m_n)}\big)^{-\frac{1}{2}}z^k,
\end{equation}
where the squared norms of the monic orthogonal polynomials are given by
\begin{equation}\label{Hk}
h_k^{(m_1,\ldots,m_n)}\equiv\int w_n^{(m_1,\ldots,m_n)}(z)|z|^{2k}d^2z=\pi^n\prod\limits_{a=1}^n\Gamma\left(k+1+m_a\right).
\end{equation}
Therefore  $K_N^{(m_1,\ldots,m_n)}(z,\zeta)$ is the kernel of the projection operator which projects from $L^2(\C,\mu)$ onto the span of $\{z^k:\;0\leq k\leq N-1\}$.
We thus conclude that the eigenvalues of $P_n$ form a determinantal projection process. We call this process  \textit{the generalised  induced finite-$N$  complex Ginibre ensemble with parameters
$m_1,m_2,\ldots,m_n$.}


\subsection{Limit of large matrices}\label{largeN}

When $N\rightarrow\infty$, then the finite-$N$ process converges (in distribution) to some limiting determinantal process, with respect to the same probability measure~\eqref{BackgroundProbabilityMeasure}.
A note of caution is in place here. As was pointed out in \cite{ABu} there are several possibilities to take the large $N$ limit. One possibility is to rescale the matrix elements of the induced Ginibre ensemble eq. (\ref{inducedGin}) to obtain an $N$ dependent distribution $\sim\exp[-N\Tr (M^{*}M)]$. In that case the limit of the spectral density $\varrho_1^{(m_1,\ldots,m_n)}(z_1)$ has compact support, and one could study local eigenvalue correlations a) at the edge of support, b) in the bulk and c) at the origin, each leading to a different limiting kernel. Because we are only interested in the hole probability we will not follow this route and keep the induced Ginibre ensemble eq. (\ref{inducedGin}) $N$ independent. In the limit $N\to\infty$ the edge of support will thus go to infinity, and without further rescaling we will always stay in case c), the origin limit.

To write the correlation kernel of this limiting determinantal point process explicitly, recall that a generalised series with an arbitrary number of numerator and denominator parameters is defined by
\[
\mathstrut_pF_q
\left(\begin{array}{cccc}
\alpha_1,  \alpha_2,  \ldots,  \alpha_p \\
\beta_1,  \beta_2,  \ldots,  \beta_q
\end{array} \biggl| z\right)
=\sum\limits_{k=0}^{\infty}
\frac{(\alpha_1)_k\cdots(\alpha_p)_k}{(\beta_1)_k\cdots(\beta_q)_k}\frac{z^k}{k!},
\]
see  Luke \cite{Luke}, Section 3.2.
The limiting determinantal process at the origin has the kernel
\begin{equation}
\begin{split}
K_{\infty}^{(m_1,\ldots,m_n)}(z,\zeta)&=\frac{1}{\pi^n}
\sum\limits_{k=0}^{\infty}\frac{(z\overline{\zeta})^k}{\prod_{a=1}^n\Gamma(k+1+m_a)}\\
&=\frac{1}{\pi^n\prod_{a=1}^n\Gamma(1+m_a)}\;
\mathstrut_1F_n\left(\begin{array}{cccc}
                        1   \\
                        1+m_1 , 1+m_2 , \ldots  , 1+m_n
                      \end{array}
\biggl|z\overline{\zeta}\right).
\end{split}
\end{equation}
We call the corresponding limiting determinantal process \textit{the generalised induced infinite complex Ginibre ensemble with parameters $m_1,m_2,\ldots,m_n$.}


\subsection{The joint density of moduli of the eigenvalues as a permanental process}
\label{SectionDistributionAbsoluteValues}

Assume that $z_1,z_2,\ldots,z_N$ are random complex variables. If their joint density with respect to the Lebesgue measure on $\C^N$ is given by
\begin{equation}\label{JointDensityUniform Order}
\varrho_N(z_1,\ldots,z_N)=\frac{1}{Z_N}\prod\limits_{k=1}^Nw(|z_k|)\prod\limits_{1\leq i<j\leq N}|z_i-z_j|^2,
\end{equation}
with the weight function $w(|z|)$ only depending on the modulus,
then the set of absolute values $|z_1|,|z_2|\ldots,|z_N|$ has the same distribution as a set of independent random variables. This is well-known fact due to Kostlan \cite{Kostlan}, which was discussed
extensively in the literature, see for example Hough, Krishnapur, Peres, and Vir\'ag \cite{Hough}, and Fyodorov and Mehlig \cite{FM} in the context of quantum-chaotic scattering. Namely, then the following result holds true:
\begin{thm}
\label{TheoremIndependence}
Assume that $z_1,z_2,\ldots,z_N$ are random complex variables. Set $z_i=r_i\e^{i\theta}$, for $i=1,\ldots,N$.  Assume that  the joint probability density function
of the (unordered) random complex variables $z_1,z_2,\ldots,z_N$ with respect to the Lebesgue measure on $\C^N$ is given by the formula
(\ref{JointDensityUniform Order})
where $w(r_k):\R_{\geq 0}\rightarrow\R_{\geq 0}$ is a normalisable weight function.
Then it holds that
\begin{equation}\label{NormaliztionConstant}
Z_N=N!\prod\limits_{k=1}^Nh_{k-1}\ ,\qquad h_k=\pi\int_0^{\infty}y^{k}w(\sqrt{y})dy,
\end{equation}
and the joint density of the (unordered) absolute values $r_1,r_2,\ldots,r_N$ is given by a permanent
\[
\prod\limits_{j=1}^N\int_0^{2\pi}d\theta_j\ \varrho_N(z_1,\ldots,z_N)
=\frac{(2\pi)^N}{Z_N}\per\left[r_i^{2j-1}\right]_{i,j=1}^N\prod\limits_{j=1}^Nw(r_j).
\]
Moreover, the set of absolute values $\{r_i\}_{i=1,\ldots,N}$ has the same distribution as the set
$\{R_i\}_{i=1,\ldots,N}$, where the random variables $R_1,R_2,\ldots R_N$ are independent, and for each $k\ (1\leq k\leq N)$ the random variable $\left(R_k\right)^2$ has the density
\[
q_k(y)=\left\{
               \begin{array}{ll}
                 \frac{y^{k-1}w(\sqrt{y})}{h_{k-1}}, & y\geq 0, \\
                 0, & y<0.
               \end{array}
             \right.
\]
\end{thm}%
In order to exploit Theorem \ref{TheoremIndependence} for the description of absolute values of eigenvalues of the  product matrix $P_n$ we need  the following result.   Recall that gamma variables $\mbox{Gamma}(k,1)$ have the following density function
\begin{equation}\label{GammaDensityFunction}
f_k(x)=\left\{
                 \begin{array}{ll}
                   \frac{1}{\Gamma(k)}x^{k-1}\e^{-x}, & x\geq 0, \\
                   0, & x<0.
                 \end{array}
               \right.
\end{equation}

\begin{prop}
\label{PropositionSpringerThompson}
Let $x_1, x_2, \ldots, x_n$ be $n$ independent gamma variables having density functions
\[
f_{b_k}(x_k)=\left\{
           \begin{array}{ll}
             \frac{1}{\Gamma(b_k)}x_k^{b_k-1}\e^{-x_k}, & x_k\geq 0, \\
             0, & x_k<0,
           \end{array}
         \right.
\]
where $b_k>0$, $1\leq k\leq n$. Then the probability density function of the product variable $z=x_1x_2\ldots x_n$ is a Meijer $G$-function multiplied
by a normalizing constant, i.e.
\[
g^{(b_1,\ldots,b_n)}(z)=\frac{G_{0,n}^{n,0}(z|b_1-1,b_2-1,\ldots,b_n-1)}{\prod_{i=1}^n\Gamma(b_i)}.
\]
\end{prop}

The proof of Proposition \ref{PropositionSpringerThompson} can be found in  Springer and Thompson \cite{Springer}, Section 3 (note however, the misprint in the denominator there).
Starting from formula (\ref{JointDensityOfEigenvalues}) for the joint density of eigenvalues of the product matrix $P_n$, and applying
Theorem \ref{TheoremIndependence} together with Proposition \ref{PropositionSpringerThompson} we obtain

\begin{thm}\label{CorollaryExplicitDistribution}
Let $P_n$  be a product of $n$ independent random matrices from the induced complex Ginibre ensemble.  The set of absolute values of eigenvalues of $P_n$ has the same distribution as a set of independent random variables
\[
\big\{R_1^{(m_1,\ldots,m_n)},R_2^{(m_1,\ldots,m_n)},\ldots,R_N^{(m_1,\ldots,m_n)}\big\},
\]
and for each $k\ (1\leq k\leq N)$ the random variable $(R_k^{(m_1,\ldots,m_n)})^2$ has the same
distribution as the product of $n$ independent gamma variables $\GAMMA(k+m_1,1)$, $\ldots$, $\GAMMA(k+m_n,1)$.
Moreover,
the random variable $(R_k^{(m_1,\ldots,m_n)})^2$  has the density function given by the formula
\begin{equation}\label{equation g_i(x)}
g_k^{(m_1,\ldots,m_n)}(x)\\
=
\begin{cases}
\dfrac{G_{0,n}^{n,0}(x\,|\, k+m_1-1,  \ldots ,  k+m_n-1 )}{\prod_{a=1}^n\Gamma(k+m_a)}, & x\geq 0, \\
0, & x<0.
\end{cases}
\end{equation}
\end{thm}
The result of Theorem \ref{CorollaryExplicitDistribution} generalises our previous results \cite{AkemannStrahov}, see Theorem 3.1 and 3.2.
There the authors considered the case when products of matrices are taken form \textit{the same} classical complex Ginibre ensemble.  Theorem \ref{CorollaryExplicitDistribution} describes distribution
of absolute values of eigenvalues  in the case when the matrices involved in the product are taken from the induced complex Ginibre ensembles with \textit{different}  parameters $m_j$.


\subsection{Hole probabilities in the generalised induced complex Ginibre ensemble}
\label{SectionHole}

Denote by $\mathcal{N}_{GG}^{(m_1,\ldots,m_n)}(r;N)$ the number of points of the generalised finite-$N$ induced complex Ginibre ensemble with  parameters $m_1, \ldots, m_n$ which are located in the disk of radius $r$ centered at the origin.
We are interested to investigate the hole probability, i.e. the probability of the event that there are no eigenvalues of $P_n$ in the disc of a radius $r$ around the origin.
We denote this probability by $\Prob\big\{\mathcal{N}_{GG}^{(m_1,\ldots,m_n)}(r;N)=0\big\}$. Starting from Theorem \ref{CorollaryExplicitDistribution} an exact expression can be given for this probability in terms of Meijer $G$-functions:
\begin{thm}\label{PropositionExactFormulae}
The hole probability $\Prob\big\{\mathcal{N}_{GG}^{(m_1,\ldots,m_n)}(r;N)=0\big\}$ for the generalised induced finite-$N$ complex Ginibre ensemble with  parameters $m_1$, $\ldots$, $m_n$ can be written as
\[
\Prob\left\{\mathcal{N}_{GG}^{(m_1,\ldots,m_n)}(r;N)=0\right\}=
\prod\limits_{k=1}^N\frac{G_{1,n+1}^{n+1,0}\left(r^2\,\bigg|\begin{array}{cccc}
                                                                  1    \\
                                                                0,  k+m_1,  \ldots,  k+m_n
                                                              \end{array}
\right)}{\prod_{a=1}^n\Gamma(k+m_a)}.
\]
\end{thm}
\begin{rem}
As for $n=1$, we can write
\[
G_{1,2}^{2,0}\left(r^2\biggl|\begin{array}{cc}
                                                                  1 \\
                                                                0,  k+m_1
                                                              \end{array}
\right)=\int_{r^2}^{\infty}\e^{-t}t^{k+m_1-1}dt=\Gamma(k+m_1,r^2).
\]
Thus we see that once $n=1$ the formula in Theorem \ref{PropositionExactFormulae} is reduced to
\[
\Prob\left\{\mathcal{N}_{GG}^{(m_1)}(r;N)=0\right\}=\prod\limits_{k=1}^{N}\frac{\Gamma(k+m_1,r^2)}{\Gamma(k+m_1)},
\]
for the induced Ginibre ensemble of a single matrix.
The formula just written above for $\Prob\big\{\mathcal{N}_{GG}^{(n=1,m_1)}(r;N)=0\big\}$ was found in  Fischmann, Bruzda,  Khoruzhenko, and Sommers \cite{Fischmann}. For $m_1=0$ it can be already found in the book of Mehta \cite{Mehta}.
\end{rem}
Theorem \ref{PropositionExactFormulae} gives  an exact formula for the hole probability in the case the products of $n$ matrices of finite size. In order to obtain the decay of the hole probability
as $r\rightarrow\infty$ we use the asymptotic expansion of the Meijer $G$-function for large values of its argument as given in Luke \cite{Luke}, Section 5.7, Theorem 5. The result is

\begin{thm}\label{TheoremFiniteNAsymptotics}
For fixed $N$ and as $r\rightarrow\infty$, we have
\begin{multline*}
\Prob\left\{\mathcal{N}_{GG}^{(m_1,\ldots,m_n)}(r;N)=0\right\} =\
\frac{(2\pi)^{\frac{N(n-1)}{2}}}{n^{\frac{N}{2}}\prod_{k=1}^N\prod_{a=1}^n\Gamma(k+m_a)}
\\
\times
\exp\left\{-nNr^{\frac{2}{n}}+N\left(N+\frac{2(m_1+\ldots+m_n)-1}{n}\right)\log r\right\}\left(1+O(r^{-\frac{2}{n}})\right).
\end{multline*}
\end{thm}

Now we turn to the case of the generalised infinite induced  Ginibre ensemble with parameters $m_1,m_2\ldots,m_n$.
Starting from Theorem \ref{CorollaryExplicitDistribution} we are able to find upper and low bounds for the hole probability
$\Prob\big\{\mathcal{N}_{GG}^{(m_1,\ldots,m_n)}(r;\infty)=0\big\}$, from which we can determine the limiting leading order bahaviour. Namely, the following result holds true.

\begin{thm}\label{TheoremBoundsForHoleProbabilities} A) (Upper bound for the hole probability.)
We have
\begin{equation}
\Prob\left\{\mathcal{N}_{GG}^{(m_1,\ldots,m_n)}(r;\infty)=0\right\}
\leq\exp\left\{-\frac{n}{4}r^{\frac{4}{n}}+n(m_1+\cdots+m_n)r^{\frac{2}{n}}+O\left(\log r\right)\right\},
\end{equation}
as $r\rightarrow\infty$.\\
B) (Lower bound for the hole probability.)
We have
\begin{multline}\label{MainLowerBoundInequality}
\Prob\left\{\mathcal{N}_{GG}^{(m_1,\ldots,m_n)}(r;\infty)=0\right\}
\geq\exp\Big\{-\frac{n}{4}r^{\frac{4}{n}}-\frac{n}{2}r^{\frac{2}{n}}\log r^{\frac{2}{n}} \\
-\sum_{a=1}^n\left(C(m_a)-m_a+\log\Gamma(m_a+1)+\log 4\right)r^{\frac{2}{n}}+O\left(\log r\right)\Big\},
\end{multline}
as $r\rightarrow\infty$, where  $C(m)$
is a constant given by
\begin{equation}\label{TheConstantC}
C(m)=m+1-\frac{1}{12(m+1)}-\left(\tfrac{1}{2}+m\right)\log(1+m)+\frac{1}{2}\int_1^{\infty}\frac{P_2(x)dx}{(m+x)^2}.
\end{equation}
Here $P_2(x)$ is the second Bernoulli periodic function
\[
P_2(x)=B_2(x-\lfloor x\rfloor).
\]
with $B_2(x)=x^2-x+\frac16$ and $\lfloor x\rfloor$ denoting the integer part of $x$. 

In particular we obtain
\begin{equation}\label{HoleLimit}
\underset{r\rightarrow\infty}{\lim}\left(\frac{1}{r^{\frac{4}{n}}}
\log\left(\Prob\left\{\mathcal{N}_{GG}^{(m_1,\ldots,m_n)}(r;\infty)\right\}=0\right)\right)=-\frac{n}{4}.
\end{equation}
\end{thm}

As we can see the first term in the asymptotic expansion of the logarithm of the hole probability does not depend on the numbers $m_1$, $\ldots$, $m_n$
characterizing the induced Ginibre ensembles. The statistical quantity defined in the left-hand of equation (\ref{HoleLimit})
depends only on the number of matrices in the product and is thus universal.
We note that in the special cases $n=1$ with square \cite{Peter} and $n=2$ with rectangular matrices \cite{APS} additional higher order terms in the asymptotic expansion were computed.


\subsection{Overcrowding estimates}\label{SectionOvercrowdingEstimates}
Let us reconsider a disk with a fixed radius $r>0$ centered  at the origin. We want to estimate the probability that there are more than $q$ points in this disc for the generalised infinite Ginibre ensemble, i.e. $\Prob\left\{\mathcal{N}^{(m_1,\ldots,m_n)}\left(r;\infty\right)\geq q\right\}$ (the probability of overcrowding). We are particularly interested in the decay of this probability when $q\rightarrow\infty$. In the context of zeros of Gaussian analytic functions the overcrowding problem was studied in Krishnapur \cite{Krishnapur}. It was shown in \cite{Krishnapur} that the probability of the event that in the disc with a fixed radius there are more than $q$ zeros of the Gaussian analytic function decays (as $q\rightarrow\infty$)  in the same way as the corresponding probability for the classical Ginibre ensemble.  The fact that the absolute values of the eigenvalues of products of random matrices are distributed as products 
of independent Gamma variables 
(see Theorem \ref{CorollaryExplicitDistribution}) enables us to get upper and lower bounds on the probabilities under considerations, and show a different behavior  of the overcrowding probabilities for products of random matrices than that in the case of the Gaussian analytic functions.
The result is

\begin{thm}\label{TheoremOvercrowdingBounds}
A) (Lower bound for the overcrowding probability.) As $q\rightarrow\infty$,
\begin{multline*}
\Prob\left\{{\mathcal{N}}_{GG}^{(m_1,\ldots,m_n)}(r;\infty)\geq q\right\}
\geq\exp\Bigl\{
-\frac{n}{2}q^2\log q +\left(\frac{n}{4}+\log r-\frac{n}{2}\log 2\right)q^2 \\
-\sum_{a=1}^n\left(m_a+\tfrac{1}{2}\right)q\log q
+O\left(q\right)\Bigr\}.
\end{multline*}
B) (Upper bound  for the overcrowding probability.) As $q\rightarrow\infty$,
\begin{multline*}
\Prob\left\{{\mathcal{N}}_{GG}^{(m_1,\ldots,m_n)}(r;\infty)\geq q\right\}
\leq \exp\Bigl\{-\frac{n}{2}q^2\log q+\Big(\frac{3n}{4}+\log r\Big)q^2 \\
+\Big(2-\sum_{a=1}^{n}m_a\Big)q\log q+O(q)\Bigr\}.
\end{multline*}
In particular, we conclude that for a fixed radius $r>0$ to leading order
\[
\Prob\left\{{\mathcal{N}}_{GG}^{(m_1,\ldots,m_n)}(r;\infty)\geq q\right\}=\exp\left\{-\frac{n}{2}q^2\log q\left(1+o(1)\right)\right\},
\]
as $q\to\infty$.
\end{thm}

Again, the first term in the asymptotic expansion of the logarithm of the overcrowding probability does not depend on the numbers $m_1,\ldots,m_n$ and is thus universal.

\newpage

\section{Products of random matrices from the induced quaternion Ginibre ensemble}
\label{SectionResultsQuaternion}

\subsection{Definition of the product matrix}
\label{SectionDefinitionOfPquaternion}

Denote by $\Hq$ the division ring of quaternions. It is known that $\Hq$ is isomorphic to the ring $\Hq'$ of $2\times 2$ matrices with complex entries
\[
\Hq'=\left\{\left(\begin{array}{cc}
                    \alpha & -\overline{\beta} \\
                    \beta & \overline{\alpha}
                  \end{array}
\right)\bigg|\,\alpha,\beta\in\C\right\},
\]
and that the elements of $\Hq'$ can be understood as matrix representations of quaternions.
Let $\mbox{Mat}\left(\Hq,N\times N\right)$ denote the collection of all $N\times N$ matrices with entries from $\Hq$, then 
the isomorphism between $\Hq$ and $\Hq'$ induces the isomorphism between $\mbox{Mat}\left(\Hq,N\times N\right)$
and
$\mbox{Mat}\left(\Hq',N\times N\right)$, where $\mbox{Mat}\left(\Hq',N\times N\right)$ is the collection of all $N\times N$ matrices whose entries are $2\times 2$
matrices from $\Hq'$.
Thus each matrix, $M\in\mbox{Mat}\left(\Hq',N\times N\right)$, can be represented as
\begin{equation}\label{M}
M=(\breve{M}_{i,j})_{i,j=1}^N
\quad\text{with}\quad
\breve{M}_{i,j}=
\begin{pmatrix}
M_{i,j}^{(1)} & -\overline{M}_{i,j}^{(2)} \\
M_{i,j}^{(2)} &  \overline{M}_{i,j}^{(1)}
\end{pmatrix},
\end{equation}
where $1\leq i, j\leq N$,  and where
$ M_{i,j}^{(1)},  M_{i,j}^{(2)}\in\C$.
Let $\mathcal{P}_{N,m}^{\quaternion}$ be a probability measure defined on $\mbox{Mat}\left(\Hq',N\times N\right)$ by
the formula
\[
\mathcal{P}_{N,m}^{\quaternion}(dM)=\frac{1}{Z_{N,m}^{\quaternion}}\left[\det M\right]^m\e^{-\frac{1}{2}\Tr M^{*}M}dM,
\]
where $m\geq 0$ is a non-negative constant, $Z_{N,m}^{\quaternion}$ is a normalization constant, and $dM$ denotes the Lebesgue measure on  $\mbox{Mat}\left(\Hq',N\times N\right)$.
Note that the eigenvalues of a matrix from $\mbox{Mat}\left(\Hq',N\times N\right)$ come in complex conjugate pairs, see Section \ref{SectionAppendixEigenvaluesQuaternionmatrices},  hence the density of $\mathcal{P}_{N,m}^{\quaternion}$ is non-negative. The probability measure $\mathcal{P}_{N,m}^{\quaternion}$  defines an ensemble of random matrices which is called the induced quaternion Ginibre ensemble (with  parameter $m$).

Let $n$ be a positive integer, $m_1,\ldots, m_n$ be positive real numbers, and consider the product $P_n^{\quaternion}$ of $n$ independent
matrices $M_1, \ldots, M_n$ from $\Mat(\Hq',N\times N)$,
\[
P_n^{\quaternion}=M_1M_2\ldots M_n.
\]
Each matrix $M_a\ (a=1,\ldots,n)$, is chosen from the induced quaternion Ginibre ensemble (with parameter $m_a$).
We will refer to $P_n^{\quaternion}$ as to the product of $n$ random matrices from the induced quaternion Ginibre ensemble with the parameters $m_1,\ldots,m_n$.


\subsection{The joint density of eigenvalues as a Pfaffian point process}
The joint density of eigenvalues of a matrix taken from the ensemble of the quaternion Ginibre matrices is well known, see Mehta \cite{Mehta}, Section 15.2. Theorem \ref{PropositionQuaternionJointDensityOfEigenvalues} generalises this result to products of independent matrices taken from the induced quaternion Ginibre ensemble with the parameters $m_1, \ldots, m_n$.

\begin{thm}\label{PropositionQuaternionJointDensityOfEigenvalues}
The vector of unordered eigenvalues of $P_n^{\quaternion}$ has a density (with respect to the Lebesgue measure on $\C^N$)
which can be written as
\begin{align}\label{JointDensityOfEigenvaluesQ}
\varrho_{N,\quaternion}^{(m_1,\ldots,m_n)}(\zeta_1,\ldots,\zeta_N)
&=\mathcal C
\prod\limits_{k=1}^Nw_n^{(m_1,\ldots,m_n)}(\zeta_k)\,|\zeta_k-\bar{\zeta}_k|^2
\prod\limits_{1\leq i<j\leq N}|\zeta_i-\zeta_j|^2|\zeta_i-\overline{\zeta}_j|^2, \\
\mathcal C^{-1} & \equiv 2^NN!\pi^{nN}\prod\limits_{k=1}^N\prod\limits_{a=1}^n\Gamma\left(m_a+2k\right), \nonumber
\end{align}
and the weight function, $w_n^{(m_1,\ldots,m_n)}(\zeta)$, is defined by equation \eqref{WeightFunction}.
\end{thm}

\begin{rem}
If $n=1$ and $m_1=0$, then formula (\ref{JointDensityOfEigenvaluesQ}) reduces to the formula for the joint density of eigenvalues of a matrix taken form the ensemble of the quaternion Ginibre matrices, see Mehta \cite{Mehta}, Section 15.2, equations 15.2.10 and 15.2.15. The case of an arbitrary $n$ and $m_1,\ldots,m_n$ was first considered in~\cite{Ipsen,IpsenKieburg},
where  a formula for the joint density was obtained without full mathematical rigor.
\end{rem}


\subsection{The correlation kernel}

\begin{defn}
Suppose that there is a $2\times 2$ matrix valued kernel $\mathbb{K}_N^{(m_1,\ldots,m_n)}(z,\zeta)$ such that for a general finitely supported function $f$ defined on $\C$ we have
\begin{equation}
\varrho_{N,\quaternion}^{(m_1,\ldots,m_n)}(\zeta_1,\ldots,\zeta_N)
\prod_{k=1}^N\left(1+f(\zeta_k)\right)
=\sqrt{\det\left(I+\mathbb{K}_N^{(m_1,\ldots,m_n)}f\right)},
\end{equation}
where $\varrho_{N,\quaternion}^{(m_1,\ldots,m_n)}(\zeta_1,\ldots,\zeta_N)$ is defined in~\eqref{JointDensityOfEigenvaluesQ}, $\mathbb{K}_N^{(m_1,\ldots,m_n)}$ is the operator associated to the kernel $\mathbb{K}_N^{(m_1,\ldots,m_n)}(z,\zeta)$, and $f$ is the operator of multiplication by the function $f$. Then $\mathbb{K}_N^{(m_1,\ldots,m_n)}(z,\zeta)$ is called the correlation kernel of the ensemble defined by (\ref{JointDensityOfEigenvaluesQ}).
\end{defn}
The explanation of the term \emph{correlation kernel} can be found in Tracy and Widom \cite{tracy}, \S 2,3.
\begin{thm}\label{TheoremMatrixValued Kernel}
The correlation kernel can be written as
\[
\mathbb{K}_N^{(m_1,\ldots,m_n)}(z,\zeta)=(\zeta-\overline{\zeta})\, w_n^{(m_1,\ldots,m_n)}(\zeta)
\begin{pmatrix}
        S_N^{(m_1,\ldots,m_n)}(\overline{z},\zeta) & -S_N^{(m_1,\ldots,m_n)}(\overline{z},\overline{\zeta}) \\
        S_N^{(m_1,\ldots,m_n)}(z,\zeta) & -S_N^{(m_1,\ldots,m_n)}(z,\overline{\zeta})
\end{pmatrix},
\]
where
\[
S_N^{(m_1,\ldots,m_n)}(z,\zeta)=\frac{1}{2\pi^n}\sum_{0\leq i\leq j\leq N-1}\prod\limits_{a=1}^n
\frac{2^{(j-i)}\Gamma\left(\frac{m_a+2j+2}{2}\right)}
{\Gamma\left(\frac{m_a+2i+2}{2}\right)\Gamma\left(m_a+2j+2\right)}\left(z^{2i}\zeta^{2j+1}-z^{2i+1}\zeta^{2j}\right).
\]
\end{thm}

\begin{rem}
1)
For $n=1$ and $m_1=0$ the formula for the correlation kernel turns into
\[
\mathbb{K}_N^{(m_1=0)}(z,\zeta)=(\zeta-\overline{\zeta})\,\e^{-|\zeta|^2}
\begin{pmatrix}
        S_N^{(m_1=0)}(\overline{z},\zeta) & -S_N^{(m_1=0)}(\overline{z},\overline{\zeta}) \\
        S_N^{(m_1=0)}(z,\zeta) & -S_N^{(m_1=0)}(z,\overline{\zeta})
\end{pmatrix},
\]
where
\[
S_N^{(m_1=0)}(z,\zeta)=\frac{1}{2\pi}\sum\limits_{0\leq i\leq j\leq N-1}2^{(j-i)}
\frac{j!}{i!(2j+1)!}\left(z^{2i}\zeta^{2j+1}-z^{2i+1}\zeta^{2j}\right).
\]
This is the well-known formula for the correlation kernel of the quaternion Ginibre ensemble, see Mehta \cite{Mehta}, equation (15.2.24).\\
2) Theorem \ref{TheoremMatrixValued Kernel} implies that the eigenvalues of $P_n^{\quaternion}$ form a Pfaffian
point process  on the complex plane with a $2\times 2$ correlation kernel. This kernel can be obtained from $\mathbb{K}_N^{(m_1,\ldots,m_n)}(z,\zeta)$
by a simple transformation.\\
3) In~\cite{Ipsen} skew-orthogonal polynomials were used to derive the correlation kernel in the case $m_1=\cdots=m_n$ and the result for arbitrary parameters was derived in~\cite{IpsenKieburg}. However, the approach given in this paper differs from the aforementioned papers and our proof of Theorem \ref{TheoremMatrixValued Kernel} does not use skew-orthogonal polynomials. Instead, we use a technique similar to that used by Tracy and Widom \cite{tracy} to derive correlation kernels for the Gaussian Orthogonal Ensemble, and for the Gaussian Symplectic Ensemble.
\end{rem}

The point process formed by eigenvalues of $P_n^{\quaternion}$ is called \textit{the generalised induced finite-$N$ quaternion Ginibre ensemble with parameters} $m_1,\ldots, m_n$.
When $N\rightarrow\infty$ (in the same sense as in Section \ref{largeN}) the point process formed by eigenvalues of $P_n^{\quaternion}$ converges to some limiting process on the complex plane.
We call this limiting process \textit{the generalised induced infinite quaternion Ginibre  ensemble with parameters} $m_1,\ldots,m_n$.


\subsection{The joint density of moduli of the eigenvalues as a permanental process}

 Theorem \ref{TheoremQIndependence} gives a similar result to that of Theorem \ref{TheoremIndependence} in the case when the joint density of $z_1,\ldots,z_N$ with respect to the Lebesgue measure on $\C^N$ has the form as in equation (\ref{JointDensityOfEigenvaluesQ}).

\begin{thm}\label{TheoremQIndependence}
Assume that $z_1,\ldots,z_N$ are random complex variables.  Assume that  the joint probability density function
of $z_1,\ldots,z_N$ with respect to the Lebesgue measure on $\C^N$ is given by the formula
\begin{equation}\label{JointDensityQuaternion}
\varrho_{N,\quaternion}(z_1,\ldots,z_N)=\frac{1}{Z_N^{\quaternion}}\prod\limits_{k=1}^Nw(|z_k|)|z_k-\overline{z}_k|^2\prod\limits_{1\leq i<j\leq N}|z_i-z_j|^2|z_i-\overline{z}_j|^2,
\end{equation}
where $w(|z_k|):\R_{\geq 0}\rightarrow\R_{\geq 0}$ is a suitable weight function.
Then
\begin{equation}\label{NormaliztionConstantQuaternion}
Z_N^{\quaternion}=N!\prod\limits_{k=1}^Nh_{k-1},
\qquad
h_k^{\quaternion}=2\pi\int_0^{\infty}y^{2k+1}w(\sqrt{y})dy,
\end{equation}
and the joint density of unordered moduli, $(r_i)_{i=1,\ldots,N}$, is
\begin{equation}\label{QRadialJointDensity}
\prod_{j=1}^N\int_0^{2\pi}d\theta_j\,\varrho_{N,\quaternion}(z_1,\ldots,z_N)=\frac{(4\pi)^N}{Z_N}\per[r_i^{4j-1}]_{i,j=1}^N\prod\limits_{j=1}^Nw(r_j).
\end{equation}
Moreover, the set of absolute values of $\{r_k\}_{k}$ has the same distribution as the set
$\{R_k\}_{k}$, where the random variables $R_1,R_2,\ldots R_N$ are independent, and for each $k\ (1\leq k\leq N$) the random variable $\left(R_k\right)^2$ has the density
\[
q_{k,\quaternion}(y)=\left\{
               \begin{array}{ll}
                 \frac{y^{2k-1}w(\sqrt{y})}{h_{k-1}^{\quaternion}}, & y\geq 0, \\
                 0, & y<0.
               \end{array}
             \right.
\]
\end{thm}

\begin{rem}
Formula (\ref{QRadialJointDensity}) was first obtained in Rider \cite{Rider}, Section 5, in the context of the classical quaternion Ginibre ensemble.
\end{rem}

As in the case of products of matrices from the induced complex  Ginibre ensemble (see Section \ref{SectionDistributionAbsoluteValues}) the application of Theorem
\ref{TheoremQIndependence} together with Proposition \ref{PropositionSpringerThompson} leads to an explicit description of the distribution of absolute values of the eigenvalues of $P_n^{\quaternion}$.

\begin{thm}\label{TheoremQ|z_1||z2|Distribution} Let $\zeta_1,\ldots,\zeta_N$ be the eigenvalues of $P_n^{\quaternion}$.
The set of absolute values of $\zeta_1, \ldots, \zeta_N$ has the same distribution as the set of independent random variables
\[
\left\{R_1^{(m_1,\ldots,m_n)},R_2^{(m_1,\ldots,m_n)},\ldots,R_N^{(m_1,\ldots,m_n)}\right\},
\]
and for each $k\ (1\leq k\leq N)$ the random variable $\big(R_k^{(m_1,\ldots,m_n)}\big)^2$ has the same
distribution as the product of $n$ independent gamma variables $\GAMMA(2k+m_1,1),\ldots,\GAMMA(2k+m_n,1)$.
Moreover, the random variable $\big(R_k^{(m_1,\ldots,m_n)}\big)^2$  has a density function given by the formula
\begin{equation}\label{equationQ g_i(x)}
g_{k,\quaternion}^{(m_1,\ldots,m_n)}(x)=
\begin{cases}
\dfrac{G_{0,n}^{n,0}\big(\,x\,\big|\,2k+m_1-1, 2k+m_2-1, \ldots, 2k+m_n-1\, \big)}{\prod_{a=1}^n\Gamma(2k+m_a)}, & x\geq 0, \\
0, & x<0.
\end{cases}
\end{equation}
\end{thm}
The result of Theorem \ref{TheoremQ|z_1||z2|Distribution} should be compared with that of Theorem \ref{CorollaryExplicitDistribution}.
The joint densities of eigenvalues of $P_n$ and $P_n^{\quaternion}$ are given by quite different formulae (see
equations (\ref{JointDensityOfEigenvalues}) and (\ref{JointDensityOfEigenvaluesQ})). In the quaternion case, equation (\ref{JointDensityOfEigenvaluesQ}) implies that the eigenvalues tend to avoid the real axis.
However, according to Theorem \ref{TheoremQ|z_1||z2|Distribution} and  Theorem \ref{CorollaryExplicitDistribution} the absolute values of eigenvalues of $P_n$ and $P_n^{\quaternion}$ are distributed in a very similar way. This is not surprising, since the different behavior in the microscopic neighbourhood of the real axis is averaged out due to the integration over the angles of the eigenvalues. For the product of quaternion Ginibre matrices this was already observed in~\cite{Ipsen}.


\subsection{Hole probabilities in the generalised induced quaternion Ginibre ensemble}

Denote by $\mathcal{N}_{GG,\quaternion}^{(m_1,\ldots,m_n)}(r;N)$ the number of eigenvalues of the random matrix $P_{n}^{\quaternion}$ in the disk of radius $r$ with its center at the origin. Similar to the case of matrices from the induced complex Ginibre ensemble, Theorem \ref{TheoremQ|z_1||z2|Distribution} enables us to give  an exact formula for the hole probability $\Prob\big\{\mathcal{N}_{GG,\quaternion}^{(m_1,\ldots,m_n)}(r;N)=0\big\}$, i.e. for the probability of the event that there are no eigenvalues of $P_n^{\quaternion}$ in the disc of a radius $r$ with its center at the origin.

\begin{thm}
\label{PropositionQuaternionExactFormulae}
The hole probability $\Prob\big\{\mathcal{N}_{GG,\quaternion}^{(m_1,\ldots,m_n)}(r;N)=0\big\}$ for the generalised induced finite-$N$  quaternion Ginibre ensemble with parameters
$m_1, \ldots, m_n$ can be written as
\[
\Prob\left\{\mathcal{N}_{GG,\quaternion}^{(m_1,\ldots,m_n)}(r;N)=0\right\}
=\prod\limits_{k=1}^N\frac{G_{1,n+1}^{n+1,0}\left(r^2\,\bigg|\begin{array}{c}
                                                                  1    \\
                                                                0,  2k+m_1,  \ldots,  2k+m_n
                                                              \end{array}
\right)}{\prod_{a=1}^n\Gamma(2k+m_a)}.
\]
\end{thm}

Theorem \ref{TheoremDecayQFiniteN} below describes the decay of the hole probability $\Prob\big\{\mathcal{N}_{GG,\quaternion}^{(m_1,\ldots,m_n)}(r;N)=0\big\}$ as $r\rightarrow\infty$.

\begin{thm}\label{TheoremDecayQFiniteN}
For fixed $N$ and $r\rightarrow\infty$, we have
\begin{multline*}\label{FiniteNQuaternionAsymptotics}
\Prob\left\{\mathcal{N}_{GG,\quaternion}^{(m_1,\ldots,m_n)}(r;N)=0\right\}=
\frac{(2\pi)^{\frac{N(n-1)}{2}}}{n^{\frac{N}{2}}\prod_{k=1}^N\prod_{a=1}^n\Gamma(2k+m_a)}\\
\times
\exp\left\{-nNr^{\frac{2}{n}}+N\left(2N+1+\frac{2(m_1+\ldots+m_n)-1}{n}\right)\log r\right\}\left(1+O\left(r^{-\frac{2}{n}}\right)\right).
\end{multline*}
\end{thm}

Comparing Theorem \ref{TheoremFiniteNAsymptotics} and Theorem \ref{TheoremDecayQFiniteN} we see that the first term in the asymptotic expansions (as $r\rightarrow\infty$)
of the logarithms of the hole probabilities at finite $N$
is the same. In other words,
\begin{equation}
\lim_{r\rightarrow\infty} \frac{\log\left[\Prob\left\{\mathcal{N}_{GG}^{(m_1,\ldots,m_n)}(r;N)=0\right\}\right]}{r^{2/n}}=
\lim_{r\rightarrow\infty} \frac{\log\left[\Prob\left\{\mathcal{N}_{GG,\quaternion}^{(m_1,\ldots,m_n)}(r;N)=0\right\}\right]}{r^{2/n}}=
-nN.
\end{equation}
Theorem \ref{TheoremDecayQ} estimates the decay of the hole probability for the generalised induced infinite quaternion Ginibre ensemble with parameters  $m_1,\ldots, m_n$.

\begin{thm}
\label{TheoremDecayQ}
Denote by $\Prob\big\{\mathcal{N}_{GG,\quaternion}^{(m_1,\ldots,m_n)}(r;\infty)=0\big\}$ the probability of the event that there are no points of the generalised induced infinite quaternion Ginibre  ensemble with parameters $m_1,\ldots, m_n$ in a disc of radius $r$ with its center at the origin.\\
A) (Upper bound for the hole probability.)
We have
\begin{equation}
\Prob\left\{\mathcal{N}_{GG;\quaternion}^{(m_1,\ldots,m_n)}(r;\infty)=0\right\}
\leq
\exp\Big\{-\frac{n}{8}r^{\frac{4}{n}}+\frac{r^{\frac{2}{n}}}{2}\Big(m_1+\cdots+m_n+\frac{n}{2}\Big)+O(\log r)\Big\},
\end{equation}
as $r\rightarrow\infty$.\\
B) (Lower bound for the hole probability.)
We have
\begin{multline}\label{MainQuaternionLowerBoundInequality}
\Prob\left\{\mathcal{N}_{GG,\quaternion}^{(m_1,\ldots,m_n)}(r;\infty)=0\right\}
\geq\exp\Bigl\{-\frac{n}{8}r^{\frac{4}{n}}
-\frac{n}{4}r^{\frac{2}{n}}\log r^{\frac{2}{n}} \\
-\frac{1}{2}\sum\limits_{a=1}^n\Big(C\Big(\frac{m_a-1}{2}\Big)
+C\Big(\frac{m_a}{2}\Big)+\log\Gamma\Big(m_a+1\Big)-m_a(1+\log 2)-\frac{1}{2}+\frac{3}{2}\log 2\Big)r^{\frac{2}{n}}
+O(\log r)\Bigr\},
\end{multline}
as $r\rightarrow\infty$, where $C(m)$ is defined by equation (\ref{TheConstantC}).

In particular this implies
\[
\underset{r\rightarrow\infty}{\lim}
\left(\frac{1}{r^{\frac{4}{n}}}\log\left[\Prob\left\{\mathcal{N}_{GG,\quaternion}^{(m_1,\ldots,m_n)}
(r;\infty)=0\right\}\right]\right)=
-\frac{n}{8}.
\]
\end{thm}

Thus for large ($N\rightarrow\infty$) matrices the hole probabilities associated with the matrices taken from induced complex Ginibre ensemble, and the hole probability associated with the matrices taken from the induced quaternion Ginibre ensemble decay almost identically. In particular, the difference in the first term in the asymptotic expansions of the logarithms of the hole probabilities is only in constant factors $-\frac{n}{4}$ and $-\frac{n}{8}$.

\begin{rem}
The result of Theorem \ref{TheoremDecayQ} is in agreement with known results for the hole probability in the case of a product of one or two matrices taken from the complex or quaternion Ginibre ensemble~\cite{APS}. Note that in~\cite{APS}, the product of two Ginibre matrices is a special case of the non-Hermitian chiral ensembles. 
\end{rem}


\subsection{Overcrowding estimates}

Overcrowding estimates for the generalised induced infinite complex Ginibre ensemble with parameters $m_1, \ldots, m_n$ (see Section \ref{SectionOvercrowdingEstimates}) can be extended to the case of the generalised induced infinite quaternion Ginibre ensemble with parameters $m_1, \ldots, m_n$. The result is given by the following

\begin{thm}\label{TheoremOvercrowdingEstimatesQ}
A) (Lower bound for the overcrowding probability.)
As $q\rightarrow\infty$,
\begin{multline*}
\Prob\left\{{\mathcal{N}}_{GG,\quaternion}^{(m_1,\ldots,m_n)}(r;\infty)\geq q\right\}
\geq\exp\Bigl\{ -nq^2\log q +\left(\frac{n}{2}-2n\log 2+2\log r\right)q^2 \\
-\sum\limits_{a=1}^n\left(m_a+1\right)q\log q
+O\left(q\right)\Bigr\}.
\end{multline*}
B) (Upper bound for the overcrowding probability.)
As $q\rightarrow\infty$,
\begin{multline*}
\Prob\left\{{\mathcal{N}}_{GG,\quaternion}^{(m_1,\ldots,m_n)}(r;\infty)\geq q\right\} 
\leq \exp\Bigl\{-nq^2\log q+\Big(\,\frac{3n}{2}-n\log 2+\log r^2\Big)q^2  \\
+\Big(2-\sum_{a=1}^n \big(m_a+\tfrac{1}{2}\big)\Big)q\log q+O(q)\Bigr\}.
\end{multline*}

Thus, for a fixed $r>0$,
\[
\Prob\left\{{\mathcal{N}}_{GG;\quaternion}^{(m_1,\ldots,m_n)}(r;\infty)\geq q\right\}=\exp\left\{-nq^2\log q(1+o(1))\right\},
\]
as $q\rightarrow\infty$. 
\end{thm}

We note that the difference between the complex and the quaternion case in the first term of asymptotic expansions as $q\rightarrow\infty$ is in constant factors $-\frac{n}{2}$ and $-n$.


\section{Joint density for the generalised complex Ginibre ensemble}
\label{SectionStartOfProofs}

We begin by completing the proof of Theorem \ref{PropositionJointDensityOfEigenvalues} as it was presented in \cite{AdhikariReddyReddySaha}. While the generalised Schur decomposition and the corresponding Jacobian were proved there, the explicit computation of the weight function and its normalisation was not spelled out explicitly.


\subsection{Proof of Theorem \ref{PropositionJointDensityOfEigenvalues}}
\label{SectionProofTheoremJointDensity}
In \cite{AdhikariReddyReddySaha} the unnormalised joint density of the eigenvalues $z_1,\ldots,z_N$ of the product matrix $P_n$
was given as
\begin{equation}
\varrho_N^{(m_1,\ldots,m_n)}(z_1,\ldots,z_N)
= \frac{1}{Z_N}
\prod\limits_{k=1}^Nw_n^{(m_1,\ldots,m_n)}(z_k)
\prod\limits_{1\leq i<j\leq N}|z_i-z_j|^2,
\end{equation}
with the weight given by
\begin{equation}\label{WeighDeltaFunction}
w_n^{(m_1,\ldots,m_n)}(\zeta)=\int\prod\limits_{a=1}^n d^2\zeta_a\, \delta^{(2)}(\zeta-\zeta_1\zeta_2\cdots \zeta_n)
\prod\limits_{a=1}^n \e^{-|\zeta_a|^2}\left|\zeta_a\right|^{2m_a},
\end{equation}
where the integration is over $\C^{nN}$. From the general theory of non-Hermitian random matrices and the corresponding orthogonal polynomials, see. e.g. \cite{Mehta}, it is clear that
\[
Z_N = N! \prod\limits_{k=1}^N h_k^{(m_1,\ldots,m_n)},
\]
where the $h_k^{(m_1,\ldots,m_n)}$ are the squared norms of the orthogonal polynomials which are monic.
Due to the delta-function in eq. (\ref{WeighDeltaFunction}) they immediately follow as the integrals factorise:
\begin{equation}
h_k^{(m_1,\ldots,m_n)}
=\prod_{a=1}^n\left (\int d^2z_a |z_a|^{2k}\e^{-|z_a|^2} \right)
=\pi^{n}\prod\limits_{a=1}^n\Gamma(k+1+m_a).
\label{normC}
\end{equation}
It remains to show that the $n$-fold integral representation of the weight function indeed lead to the Meijer $G$-function as claimed. The proof follows the one for the product of non-induced Ginibre matrices \cite{ABu}
very closely.
First, it is easy to see that
equation (\ref{WeighDeltaFunction}) implies
\begin{equation}\label{WeightFunctionIntegralI}
w_n^{(m_1,\ldots,m_n)}(\zeta)=(2\pi)^{n-1}|\zeta|^{2m_n} 
\prod_{i=1}^{n-1}\int_0^\infty\frac{dr_i}{r_i}r_i^{2(m_i-m_n)}
\exp\biggl\{-\sum\limits_{j=1}^{n-1}r_j^2-\frac{|\zeta|^2}{r_1^2\cdots r_{n-1}^2}\biggr\}.
\end{equation}
Indeed, to obtain formula (\ref{WeightFunctionIntegralI}) from equation (\ref{WeighDeltaFunction}) we first integrate over $z_n$, then we introduce the polar coordinates
\[
z_1=r_1\e^{i\varphi_1},\ldots,z_{n-1}=r_{n-1}\e^{i\varphi_{n-1}},
\]
and integrate over the angles $\varphi_1,\ldots,\varphi_{n-1}$. Starting from equation (\ref{WeightFunctionIntegralI}) and in analogy to \cite{ABu} the following recurrence relation holds
\begin{equation}\label{WeightFunctionIntegralIII}
w_{n+1}^{(m_1,\ldots,m_{n+1})}(\zeta)=2\pi\int_0^{\infty}
\frac{dr_n}{r_n}w_{n}^{(m_1,\ldots,m_{n-1},m_{n+1})}\left(\frac{\zeta}{r_n}\right)\e^{-r_n^2}r_n^{2m_n}.
\end{equation}
Note that from equation (\ref{WeightFunctionIntegralI}) it follows that the weight function $w_n^{(m_1,\ldots,m_n)}(\zeta)$ depends only on $|\zeta|$. Taking this into account, we define the following function of the squared radius of the eigenvalue, $R\equiv r^2=|\zeta|^2$,
\[
\Omega^{(m_1,\ldots,m_n)}_n(R)=w_n^{(m_1,\ldots,m_n)}(\sqrt{R}).
\]
Equation(\ref{WeightFunctionIntegralIII}) implies that
\begin{equation}\label{WeightFunctionIntegralIV}
\Omega_{n+1}^{(m_1,\ldots,m_{n+1})}(R)=\pi\int_0^{\infty}
d\rho\, \Omega_{n}^{(m_1,\ldots,m_{n-1},m_{n+1})}\left(\frac{R}{\rho}\right)\e^{-\rho}\rho^{m_n-1}.
\end{equation}
Let $M^{(m_1,\ldots,m_n)}_n(s)$ be the Mellin transform of $\Omega^{(m_1,\ldots,m_n)}_n(R)$,
\begin{equation}\label{WeightFunctionIntegralV}
M^{(m_1,\ldots,m_n)}_n(s)=\int_0^{\infty}dR\,R^{s-1}\Omega^{(m_1,\ldots,m_n)}_n(R).
\end{equation}
It follows from equations (\ref{WeightFunctionIntegralIV}) and (\ref{WeightFunctionIntegralV}) that
\begin{equation}\label{WeightFunctionIntegralVI}
M^{(m_1,\ldots,m_{n+1})}_{n+1}(s)=\pi\Gamma(m_n+s)M^{(m_1,\ldots,m_{n-1},m_{n+1})}_n(s).
\end{equation}
Moreover, taking into account that $w_1^{(m_1)}(\zeta)=|\zeta|^{2m_1}\e^{-|\zeta|^2}$, we obtain
the initial condition
\begin{equation}\label{WeightFunctionIntegralVII}
M^{(m_1)}_{1}(s)=\Gamma(m_1+s).
\end{equation}
From equation (\ref{WeightFunctionIntegralVI}), and equation (\ref{WeightFunctionIntegralVII}) we find that
\[
M^{(m_1,\ldots,m_n)}_n(s)=\pi^{n-1}\prod\limits_{j=1}^n\Gamma(m_j+s).
\]
The inverse Mellin transform of $M^{(m_1,\ldots,m_n)}_n(s)$ is
\[
\Omega^{(m_1,\ldots,m_n)}_n(R)=\pi^{n-1}\frac{1}{2\pi i}\int\limits_{c-i\infty}^{c+i\infty}\bigg(\prod\limits_{j=1}^n\Gamma(m_j+s)\bigg)R^{-s}ds,
\]
where $c$ is an arbitrary strictly positive parameter. We conclude that
\[
\Omega^{(m_1,\ldots,m_n)}_n(R)=\pi^{n-1}G_{0,n}^{n,0}(R\, |\,m_1,  m_2,  \ldots,  m_n ),
\]
and that formula (\ref{WeightFunction}) for the weight function holds true.
Theorem \ref{PropositionJointDensityOfEigenvalues} is proved.
\qed


\subsection{Proof of Theorem \ref{TheoremDeterminantalCorrelationkernel}}
Theorem \ref{TheoremDeterminantalCorrelationkernel} can be obtained from  Theorem \ref{PropositionJointDensityOfEigenvalues}
by the same method as in the case of the classical complex Ginibre ensemble, see Mehta \cite{Mehta}, Section 15.1.
Indeed, equation (\ref{JointDensityOfEigenvalues}) implies that the eigenvalues of $P_n$ form a determinantal point process on the complex plane. The corresponding correlation kernel
has the following form
\[
K_N^{(m_1,\ldots,m_n)}(z,\zeta)=\sum\limits_{k=0}^{N-1}\varphi_k(z)\overline{\varphi_k(\zeta)},
\]
where $\varphi_k(z)$ are polynomials orthonormal with respect to the weight function $w_n^{(m_1,\ldots,m_n)}(z)$,
\[
\int \varphi_k(z)\overline{\varphi_l(z)}w_n^{(m_1,\ldots,m_n)}(z)d^2z=\delta_{k,l}.
\]
Because the weight function is rotational invariant these polynomials are monic and can be properly normalised as
\[
\varphi_k(z)=\frac{z^k}{h_k^{(m_1,\ldots,m_n)}},
\]
where the squared norms $h_k^{(m_1,\ldots,m_n)}$ were already computed as the moments in equation (\ref{normC}).
Theorem \ref{TheoremDeterminantalCorrelationkernel} follows.
\qed


We will now study the radial distributions of the eigenvalues of $P_n$, and of the points of the generalised induced complex Ginibre ensemble with parameters $m_1,\ldots,m_n$. The first step in this direction is to prove Theorem \ref{TheoremIndependence} on the distribution of the absolute values of the complex variables $z_1, \ldots, z_N$
whose joint density has the same form as that of  the joint density of the eigenvalues of $P_n$. This follows along the lines of the work of Kostlan \cite{Kostlan} for Gaussian weights, and ultimately leads to the independence of the moduli of the complex eigenvalues.


\subsection{Proof of Theorem \ref{TheoremIndependence}}Set
$
z_i=r_i\e^{i\theta_i}\ (i=1,\ldots, N)
$.
Expanding the two Vandermonde determinants in the joint density we have
\begin{align*}
\prod\limits_{1\leq i<j\leq N}|z_i-z_j|^2
&=\bigg|\sum\limits_{\sigma\in S(N)}(-1)^{\sgn(\sigma)}\prod\limits_{j=1}^Nz_j^{\sigma(j)-1}\bigg|^2\\
&=\sum\limits_{\sigma,\sigma'\in S(N)}(-1)^{\sgn(\sigma)+\sgn(\sigma')}
\prod\limits_{j=1}^Nr_j^{\sigma(j)-1}\e^{i(\sigma(j)-1)\theta_j}
\prod\limits_{k=1}^Nr_k^{\sigma'(k)-1}\e^{-i(\sigma'(k)-1)\theta_k},
\end{align*}
where the summations are over the permutation group, $S(N)$. This implies
\[
\prod\limits_{j=1}^N\int_0^{2\pi}d\theta_j
\prod\limits_{1\leq i<j\leq N}|z_i-z_j|^2=(2\pi)^N\sum\limits_{\sigma\in S(N)}\prod\limits_{j=1}^Nr_j^{2\sigma(j)-2}.
\]
Therefore the joint density of the unordered $r_1,\ldots,r_N$  is given by a permanent
\begin{align*}
\prod\limits_{j=1}^N\int_0^{2\pi}d\theta_j\, \varrho_N^{(m_1,\ldots,m_n)}(z_1,\ldots,z_N)
&=\frac{(2\pi)^N}{Z_N}\prod\limits_{k=1}^Nw(r_k)\sum\limits_{\sigma\in S(N)}\prod\limits_{j=1}^Nr_j^{2\sigma(j)-2}\prod\limits_{j=1}^Nr_j \\
&=\frac{1}{N!}\per\left[\frac{r_i^{2j-1}w(r_i)}{\int_0^{\infty}r^{2j-1}w(r)dr}\right]_{i,j=1}^N.
\end{align*}
From this expression for the joint density the statement of Theorem \ref{TheoremIndependence} follows immediately. Moreover, in the last step we have included the normalisation inside the permanent. Therefore
the vector of squares of  $r_1\ldots,r_N$ has the density
\[
\frac{1}{N!}\per\left[q_j(y_i)\right]_{i,j=1}^N,
\]
where the functions $q_j(y)$, $1\leq j\leq N$, are probability density functions defined by
\[
q_j(y)=\left\{
               \begin{array}{ll}
                 \frac{y^{j-1}w(\sqrt{y})}{h_{j-1}}, & y\geq 0 \\
                 0, & y<0.
               \end{array}
             \right.
\]
To complete the proof of Theorem \ref{TheoremIndependence} we use the following well known fact (see, for example, Kostlan \cite{Kostlan}, Lemma 1.5).
Assume we are given an $N$-tuplet of independent random variables $A_i\ (1\leq i\leq N)$, with the corresponding densities $q_i\ (1\leq i\leq N)$.
Define a new $N$-tuplet of random variables, $B_i\ (1\leq i\leq N)$, as a random permutation of the vector $A_i\ (1\leq i\leq N)$ (these random permutations are equal
to each other in probability).  Then the joint density of the random vector $B_i\ (1\leq i\leq N)$, is
\[
\frac{1}{N!}\per\left[q_j(B_i)\right]_{i,j=1}^N.
\]
Considering squares of moduli of unordered variables  $z_1,\ldots,z_N$  as random variables $B_i\ (1\leq i\leq N)$, we obtain the statement of Theorem \ref{TheoremIndependence}.
\qed


\subsection{Proof of Theorem \ref{CorollaryExplicitDistribution}}

Let $z_1,\ldots, z_N$ be the eigenvalues of $P_n$. From Proposition \ref{JointDensityOfEigenvalues} and Theorem \ref{TheoremIndependence} we conclude that
the  set of absolute values of $z_1,\ldots, z_N$ has the same distribution as the set of independent random variables,
\[
\left\{R_1^{(m_1,\ldots,m_n)},R_2^{(m_1,\ldots,m_n)},\ldots,R_N^{(m_1,\ldots,m_n)}\right\}.
\]
For each $k$ ($1\leq k\leq N$), the random variable $\big(R_k^{(m_1,\ldots,m_n)}\big)^2$ has the density
\begin{equation}
q_k^{(m_1,\ldots,m_n)}(x)=
\begin{cases}
\dfrac{x^{k-1}\pi^{k-1}G_{0,n}^{n,0}(x\,|\,m_1,  m_2,  \ldots,  m_n )}
{\int_{0}^{\infty}x^{k-1}\pi^{k-1}G_{0,n}^{n,0}(x\,|\, m_1,  m_2,  \ldots,  m_n )dx}, & x\geq 0, \\
                   0, & x<0.
\end{cases}
\end{equation}

This gives formula (\ref{equation g_i(x)}) for the density function. Now, formula (\ref{equation g_i(x)}) and Proposition \ref{PropositionSpringerThompson}
imply that $\big(R_k^{(m_1,\ldots,m_n)}\big)^2$ has the same
distribution as the product of $n$ independent gamma random variables $\GAMMA(k+m_1,1),\ldots,\GAMMA(k+m_n,1)$.
\qed


\section{Hole probabilities for the generalised complex Ginibre ensemble}
\label{SectionProofsHole}

This Section contains the proofs of exact and asymptotic results for the hole probabilities both in the case of the generalised induced finite-$N$ and infinite complex Ginibre ensemble with parameters $m_1,\ldots,m_n$.


\subsection{Proof of Theorem \ref{PropositionExactFormulae}}
Recall that  $\mathcal{N}_{GG}^{(m_1,\ldots,m_n)}(r;N)$ denotes the number of points of the generalised induced finite-$N$ complex Ginibre ensemble with  parameters $m_1, \ldots, m_n$ in the disk of radius $r$ with its center at the origin. From Theorem \ref{CorollaryExplicitDistribution} we conclude that
the hole probability can be written as
\begin{equation}
\Prob\left\{\mathcal{N}_{GG}^{(m_1,\ldots,m_n)}(r;N)=0\right\}
=\prod\limits_{k=1}^N\Prob\left\{\left(R_k^{(m_1,\ldots,m_n)}\right)^2>r^2\right\},
\end{equation}
where the random variables $R_k^{(m_1,\ldots,m_n)}$ are those defined in the statement of Theorem \ref{CorollaryExplicitDistribution}. Formula
(\ref{equation g_i(x)}) gives
\[
\Prob\left\{\mathcal{N}_{GG}^{(m_1,\ldots,m_n)}(r;N)=0\right\}\\
=\prod_{k=1}^N
\frac{\displaystyle{\int_{r^2}^{\infty}}G_{0,n}^{n,0}(x\,|\,k+m_1-1, k+m_2-1, \ldots, k+m_n-1)dx}{\prod_{a=1}^n\Gamma(k+m_a)}.
\]
The integral just written above can be computed explicitly, see Luke \cite{Luke}, Section 5.6.  As a result we obtain the formula for the hole probabilities in the statement of Theorem \ref{PropositionExactFormulae}. \qed

\subsection{Proof of Theorem \ref{TheoremFiniteNAsymptotics}}

We use an asymptotic formula for Meijer $G$-functions (see Luke \cite{Luke}, Section 5.7)
\begin{equation}\label{MeijerGAsymptotics}
G_{p,q}^{q,0}\left(x\,\bigg|
\begin{matrix}
 a_1,  a_2,  \ldots ,  a_p \\
 b_1,  b_2,  \ldots ,  b_q
\end{matrix}\right)
\sim\frac{(2\pi)^{(\sigma-1)/2}}{\sigma^{1/2}}\exp\left\{-\sigma x^{1/\sigma}\right\}x^{\theta}\sum\limits_{k=0}^{\infty}M_k x^{-k/\sigma},
\end{equation}
for $x\to\infty$. The parameters $\sigma$ and $\theta$ are defined by
\[
\sigma=q-p
\qquad\text{and}\qquad
\sigma\theta=\frac{1}{2}(1-\sigma)+\sum\limits_{k=1}^qb_k-\sum\limits_{k=1}^pa_k.
\]
The coefficients $M_k$ are independent of $x$ and are given in the above reference. However, we only need that $M_0=1$. 
This leads to the asymptotic expression
\[
G_{1,n+1}^{n+1,0}\left(r^2\,\bigg|
\begin{matrix}
1 \\
0,  k+m_1,  \ldots,  k+m_n
\end{matrix}\right)
=\frac{(2\pi)^{\frac{n-1}{2}}}{n^{\frac{1}{2}}}\e^{-nr^{\frac{2}{n}}}
r^{2k-1+\frac{2}{n}(m_1+\cdots+m_n)-\frac{1}{n}}
\left[1+O\left(r^{-2/n}\right)\right]
\]
as $r\rightarrow\infty$. When inserting this into the formula for the hole probability in Theorem \ref{PropositionExactFormulae},
we obtain the asymptotic formula in Theorem \ref{TheoremFiniteNAsymptotics}. \qed


\subsection{Proof of Theorem \ref{TheoremBoundsForHoleProbabilities}}

\subsubsection{Upper bound for the hole probability of the infinite ensemble}

To estimate the hole probability at $N\to\infty$ we use the following standard fact called the Markov inequality.
\begin{prop}\label{PropositionMarkovInequality}
Suppose $\varphi:\R\rightarrow\R$ is a positive valued function, and let $A$ be a Borel subset of $\R$. Then
\[
\inf\left\{\varphi(y):y\in A\right\}\Prob\left\{X\in A\right\}\leq \mathbb{E}\{\varphi(X)\}.
\]
\end{prop}
\begin{proof}
See, for example, Durrett \cite{Durrett}, Section 1.6, Theorem 1.6.4.
\end{proof}
Let $b\geq 0$, $A=(r^2,\infty)$, and $\varphi(x)=x^{b}$. By the Markov inequality (Proposition \ref{PropositionMarkovInequality})
we have%
\footnote{In the proof for the $\beta=4$ case one should replace $k$ with $2k$.}
\[
\Prob\left\{\left(R_k^{(m_1,\ldots,m_n)}\right)^2>r^2\right\}
\leq\frac{\displaystyle{\int_{0}^{\infty}}t^{b}G_{0,n}^{n,0}(t\,|\,k+m_1-1,k+m_2-1,\ldots,k+m_n-1)dt}
{(r^2)^{b}\prod_{a=1}^n\Gamma(k+m_a)}.
\]
Moreover,
\[
\int_{0}^{\infty}t^{b}G_{0,n}^{n,0}(t\,|\,k+m_1-1,k+m_2-1,\ldots,k+m_n-1)dt=\prod_{a=1}^n\Gamma(k+m_a+b).
\]
Therefore,
\[
\Prob\left\{\left(R_k^{(m_1,\ldots,m_n)}\right)^2>r^2\right\}
\leq\frac{1}{(r^2)^{b}}\prod_{a=1}^n\frac{\Gamma(k+m_a+b)}{\Gamma(k+m_a)}.
\]
As the next step we utilise the following well-known inequality (see e.g. the Digital Library of Mathematical Functions \cite{Digital}, $\S$ 5.6)
\begin{equation}\label{GammaInequality}
(2\pi)^{\frac{1}{2}}\exp\left\{-z+(z-\frac{1}{2})\log z\right\}\leq\Gamma(z)\leq(2\pi)^{\frac{1}{2}}\exp\left\{-z+(z-\tfrac{1}{2})\log(z)+\frac{1}{12 z}\right\},
\end{equation}
as $z\geq 1$,
to obtain
\begin{multline*}
\Prob\left\{\left(R_k^{(m_1,\ldots,m_n)}\right)^2>r^2\right\}
\leq\exp\biggl\{-nb-b\log(r^2)+\sum\limits_{a=1}^n\left(k+m_a-\tfrac{1}{2}+b\right)\log\left(k+m_a+b\right)\\
-\sum\limits_{a=1}^n(k+m_a-\tfrac{1}{2})\log(k+m_a)+\sum\limits_{a=1}^n\frac{1}{12(k+m_a+b)}\biggr\}.
\end{multline*}
Now set $b=r^{\frac{2}{n}}-k$ in the equation above\footnote{In the proof for the $\beta=4$ case one should choose $b=r^{\frac{2}{n}}-2k$.}. We can then rewrite the inequality above as
\begin{multline*}
\Prob\left\{\left(R_k^{(m_1,\ldots,m_n)}\right)^2>r^2\right\}\leq
\exp\biggl\{-nr^{\frac{2}{n}}-nr^{\frac{2}{n}}\log r^{\frac{2}{n}}+\sum\limits_{a=1}^n(r^{\frac{2}{n}}+m_a-\tfrac{1}{2})\log(r^{\frac{2}{n}}+m_a)\\
  +nk\big(1+\log r^{\frac{2}{n}}\big)-\sum\limits_{a=1}^n(m_a-\tfrac{1}{2})\log\left(k+m_a\right)-k\sum\limits_{a=1}^n\log\left(k+m_a\right)
+\frac{1}{12}\sum\limits_{a=1}^n\frac{1}{r^{\frac{2}{n}}+m_a}\biggr\}.
\end{multline*}
Next we exploit that\footnote{For $\beta=4$ one should replace $r^{\frac{2}{n}}$ with $r^{\frac{2}{n}}/2$ as the upper limit for the product.}
\[
\prod_{k=1}^{\infty}\Prob\left\{\left(R_k^{(m_1,\ldots,m_n)}\right)^2>r^2\right\}\leq \prod\limits_{k=1}^{r^{\frac{2}{n}}}\Prob\left\{\left(R_k^{(m_1,\ldots,m_n)}\right)^2>r^2\right\}.
\]
This gives
\begin{multline*}
\prod_{k=1}^{\infty}\Prob\left\{\left(R_k^{(m_1,\ldots,m_n)}\right)^2>r^2\right\}\leq\\
\exp\biggl\{-nr^{\frac{4}{n}}-nr^{\frac{4}{n}}\log r^{\frac{2}{n}}+r^{\frac{2}{n}}
\sum_{a=1}^n(r^{\frac{2}{n}}+m_a-\tfrac{1}{2})\log(r^{\frac{2}{n}}+m_a)
+\frac{n}{2}r^{\frac{2}{n}}(r^{\frac{2}{n}}+1) (1+\log r^{\frac{2}{n}})    \\
+\frac{1}{2}\sum\limits_{a=1}^n\sum\limits_{k=1}^{r^{\frac{2}{n}}}\log\left(k+m_a\right)
-\sum\limits_{a=1}^n\sum\limits_{k=1}^{r^{\frac{2}{n}}}(k+m_a)\log\left(k+m_a\right)
+\frac{r^{\frac{2}{n}}}{12}\sum\limits_{a=1}^n\frac{1}{r^{\frac{2}{n}}+m_a}
\biggr\}.
\end{multline*}
Here $r^{\frac{2}{n}}$ is considered as an integer (this assumption should not affect our estimate), alternatively we simply consider $\lfloor r^{\frac{2}{n}}\rfloor+1$.
A straightforward application of equations (\ref{AsymptoticsSum1}) and (\ref{AsymptoticsSum2}) from the Appendix \ref{App} gives
\begin{align*}
\sum\limits_{k=1}^{r^{\frac{2}{n}}}\log\left(k+m_a\right)&=r^{\frac{2}{n}}\log r^{\frac{2}{n}}-r^{\frac{2}{n}}+O\left(\log r\right), \\
\sum\limits_{k=1}^{r^{\frac{2}{n}}}(k+m_a)\log\left(k+m_a\right)&=\frac{1}{2}r^{\frac{4}{n}}\log r^{\frac{2}{n}}-\frac{r^{\frac{4}{n}}}{4}
+r^{\frac{2}{n}}\left(m_a+\tfrac{1}{2}\right)\log r^{\frac{2}{n}}+O\left(\log r\right).
\end{align*}
In addition,
\[
r^{\frac{2}{n}}\sum_{a=1}^n(r^{\frac{2}{n}}+m_a-\tfrac{1}{2})\log(r^{\frac{2}{n}}+m_a)
=nr^{\frac{4}{n}}\log r^{\frac{2}{n}}+r^{\frac{2}{n}}\log r^{\frac{2}{n}}\sum\limits_{a=1}^n(m_a-\tfrac{1}{2})
+r^{\frac{2}{n}}\sum\limits_{a=1}^nm_a+O(1).
\]
The estimate in the statement of  Theorem \ref{TheoremBoundsForHoleProbabilities}, (A) follows. \qed


\subsubsection{Lower bound for the hole probability for the infinite ensemble}
\label{SectionLowerBound}

\begin{prop}
\label{PropositionLowerBound1}
We have
\begin{multline}
\prod\limits_{k=1}^{\infty}\Prob\Big\{\Big(R_k^{(m)}\Big)^2>r^2\Big\}
\geq\exp\Bigl\{-\frac{r^4}{4}-\frac{r^2\log r^{2}}{2} \\
-r^{2}\left(C(m)-m+\log(m!)+\log 4\right)+O\left(\log(r)\right)\Bigr\},
\end{multline}
as $r\rightarrow\infty$. Here the constant $C(m)$ is defined by equation (\ref{TheConstantC}).
\end{prop}

\begin{proof}
The idea of the proof is to split the infinite product into three parts and to derive bounds for each of these.

Recall that $\big(R_k^{(m)}\big)^2$ is distributed in the same way as a gamma random variable\footnote{In the proof for the $\beta=4$ case we have $\GAMMA(2k+m,1)$.},  $\GAMMA(k+m,1)$.
Following Hough, Krishnapur, Peres, and Vir\'ag \cite{Hough}, Section 7.2, we use the fact that
\[
\Prob\left\{\GAMMA(k+m,1)>\lambda\right\}=\Prob\left\{\Poisson(\lambda)<k+m\right\}.
\]
This gives
\begin{equation}
\Prob\left\{\left(R_k^{(m)}\right)^2>r^2\right\}=\Prob\left\{\Poisson(r^2)\leq k+m-1\right\}\geq\exp\left\{-r^2\right\}\frac{r^{2(k+m-1)}}{(k+m-1)!},
\end{equation}
and thus for the first part of the infinite product we have\footnote{In the proof for $\beta=4$ we replace $k$ with $2k$ and change the upper limit of the product from $r^2$ to $\frac{r^2}{2}$.}
\begin{align*}
\prod\limits_{k=1}^{r^2}\Prob\left\{\left(R_k^{(m)}\right)^2>r^2\right\}
&\geq\prod\limits_{k=1}^{r^2}\exp\left\{-r^2\right\}\frac{r^{2(k+m-1)}}{(k+m-1)!} \\
&=\exp\Big\{-r^4+\log r^2\sum\limits_{k=1}^{r^2}(k+m-1)-\sum\limits_{k=1}^{r^2}\log(k+m-1)!\Big\}.
\end{align*}
Without loss of generality we have assumed that $r^2$ is an integer. The following identity holds true
\begin{equation}
\begin{split}
\sum\limits_{k=1}^{r^2}\log(k+m-1)!
&=r^2\log m!+(r^2+m)\sum\limits_{k=1}^{r^2}\log(k+m)-\sum\limits_{k=1}^{r^2}(k+m)\log(k+m).
\end{split}
\nonumber
\end{equation}
Applying formulae (\ref{AsymptoticsSum1}) and (\ref{AsymptoticsSum2}) from  Appendix~\ref{App} we obtain
\begin{equation}\label{FirstTermsInequality}
\begin{split}
\prod_{k=1}^{r^2}\Prob\Big\{\Big(R_k^{(m)}\Big)^2\!>r^2\Big\}
\geq
\exp\Big\{-\frac{r^4}{4}-\frac{r^2}{2}\log r^2-r^2(C(m)-m+\log m!)+O(\log(r))\Big\}
\end{split}
\end{equation}
as $r\rightarrow\infty$, where $C(m)$ is defined  by equation (\ref{TheConstantC}).

We move to an estimate for the second part of the product. Since $\big(R_k^{(m)}\big)^2$ is distributed in the same way as $\GAMMA(k+m,1)$, we have
\begin{align*}
\Prob\Big\{\left(R_k^{(m)}\right)^2>r^2\Big\}&=\Prob\left\{\Poisson(r^2)\leq k+m-1\right\}
=1-\Prob\left\{\Poisson(r^2)> k+m-1\right\} \\
&\geq 1-\Prob\left\{\Poisson(r^2)\geq r^2\right\},
\end{align*}
for $k>r^2-m$. Since
$\Prob\left\{\Poisson(r^2)\geq r^2\right\}\rightarrow\frac{1}{2}$ as $r\rightarrow\infty$,
we obtain that
\begin{equation}
\Prob\Big\{\Big(R_k^{(m)}\Big)^2>r^2\Big\}\geq\frac{1}{4},
\nonumber
\end{equation}
for sufficiently large $r$ and for $k\geq r^2$. This gives\footnote{For $\beta=4$ this product goes from $\frac{r^2}{2}+1$ to $r^2$, and obviously in the third part the product extends then from $r^2+1$ to $\infty$.}
\begin{equation}\label{SecondTermsInequality}
\prod_{k=r^2+1}^{2r^2}\Prob\Big\{\Big(R_k^{(m)}\Big)^2>r^2\Big\}\geq\exp\left\{-r^2\log 4\right\},
\end{equation}
for sufficiently large $r$ which was already derived in \cite{Hough}.

For the third part of the product we once more follow the same argument as in Hough,  Krishnapur,  Peres and Vir$\acute{\mbox{a}}$g \cite{Hough}. Section 7.2 there gives
\[
\Prob\left\{\left(R_k^{(m)}\right)^2<\frac{k+m}{2}\right\}\leq\exp\left\{-c(k+m)\right\},
\]
where the constant $c$ is independent of $k$. This implies
\[
\sum\limits_{k=2r^2+1}^{\infty}\Prob\left\{\left(R_k^{(m)}\right)^2<r^2\right\}
\leq\sum\limits_{k=2r^2+1}^{\infty}\Prob\left\{\left(R_k^{(m)}\right)^2<\frac{k+m}{2}\right\}<\frac{1}{2},
\]
for sufficiently large $r$. Finally, since
\[
\prod\limits_{j=1}^\infty(1-a_j)\geq 1-\sum\limits_{i=1}^{\infty}a_i,
\]
where $0<a_i<1$ are such that the product in the left-hand side, and the sum on the right-hand side converge, we have
\begin{equation}\label{ThirdTermsInequality}
\begin{split}
\prod\limits_{k=2r^2+1}^{\infty}\Prob\left\{\left(R_k^{(m)}\right)^2>r^2\right\}
&=\prod\limits_{k=2r^2+1}^{\infty}\left(1-\Prob\left\{\left(R_k^{(m)}\right)^2\leq r^2\right\}\right)\\
&\geq 1-\sum\limits_{k=2r^2+1}^{\infty}\Prob\left\{\left(R_k^{(m)}\right)^2<\frac{k+m}{2}\right\}>\frac{1}{2}.
\end{split}
\end{equation}
The statement of the Proposition follows from inequalities (\ref{FirstTermsInequality}), (\ref{SecondTermsInequality}), and (\ref{ThirdTermsInequality}) for the different parts of the product.
\end{proof}
We  know from Theorem~\ref{CorollaryExplicitDistribution} that the random variables $R_k^{(m_1,\ldots,m_n)}$ are independent, and that each 
$\big(R_k^{(m_1,\ldots,m_n)}\big)^2$ enjoys the same distribution as the product of $n$ independent gamma random variables $\GAMMA(k+m_1,1),\ldots,\GAMMA(k+m_n,1)$.
In particular, for every $a$, $1\leq a\leq n$, the random variable $\big(R_k^{(m_a)}\big)^2$ is itself a gamma random variable $\GAMMA(k+m_a,1)$. We conclude that
the random variable $\big(R_k^{(m_1,\ldots,m_n)}\big)^2$ has the same distribution as the product of random variables $\big(R_k^{(m_1)}\big)^2,
\big(R_k^{(m_2)}\big)^2, \ldots, \big(R_k^{(m_n)}\big)^2$.
Thus we have
\begin{equation}\label{LowerBoundInequality}
\Prob\Big\{\Big(R_k^{(m_1,\ldots,m_n)}\Big)^2>r^2\Big\}
\geq\prod_{i=1}^n\Prob\Big\{\Big(R_k^{(m_i)}\Big)^2>r^\frac{2}{n}\Big\}.
\end{equation}
Inequality (\ref{LowerBoundInequality}) and Proposition \ref{PropositionLowerBound1} imply  inequality (\ref{MainLowerBoundInequality}).
\qed


\section{Overcrowding probability for the generalised complex Ginibre ensemble}
\label{SectionProofsOvercrowding}

In this Section we continue to study probabilistic quantities of interest for the generalised induced complex Ginibre ensemble with parameters $m_1, \ldots, m_n$. Namely, we now turn to the proof of Theorem \ref{TheoremOvercrowdingBounds} giving lower and upper bounds for the overcrowding probabilities. As in the case of the hole probabilities, Theorem \ref{CorollaryExplicitDistribution} serves as the main technical tool in our analysis.  We use a technique similar to that used by Krishnapur \cite{Krishnapur} to estimate hole probabilities in the case of the classical infinite Ginibre ensemble.


\subsection{Proof of Theorem \ref{TheoremOvercrowdingBounds} A) Lower bound for the overcrowding probability.}

Let us recall that the set of absolute values of eigenvalues of this ensemble has the same distribution as the set of independent random variables
\[
\left\{R_1^{(m_1,\ldots,m_n)},R_2^{(m_1,\ldots,m_n)},\ldots,R_N^{(m_1,\ldots,m_n)}\right\}.
\]
Furthermore, the random variable $\big(R_k^{(m_1,\ldots,m_n)}\big)^2$ has the same
distribution as the product of $n$ independent gamma random variables $\GAMMA(k+m_1,1)$, $\ldots$, $\GAMMA(k+m_n,1)$.
This implies
\[
\left(R_{k}^{(m_1,\ldots,m_n)}\right)^2
\overset{d}{=}\prod_{i=1}^n\GAMMA(k+m_i,1)\\
\overset{d}{=}\prod_{i=1}^n \left(\xi^{(i)}_1+\ldots+\xi_{k+m_i}^{(i)}\right),
\]
where all $\xi_i^{(j)}$ are i.i.d. exponential random variables with mean $1$, and $\overset{d}{=}$ denotes equality in distribution. Hence we can write
\[
\Prob\Big\{\Big(R_k^{(m_1,\ldots,m_n)}\Big)^2<r^2\Big\}
\geq\prod_{i=1}^n\Prob\{\xi_1^{(i)}+\cdots+\xi_{k+m_i}^{(i)}<r^{\frac{2}{n}}\} 
\geq\prod_{i=1}^n\prod_{j=1}^{k+m_i}\Prob\Big\{\xi_j^{(i)}<\frac{r^{\frac{2}{n}}}{k+m_i}\Big\}.
\]
For any exponential random variable $\xi$  with mean $1$ it holds
\[
\Prob\left\{\xi<x\right\}\geq\frac{x}{2},\qquad 0<x<1.
\]
This leads to
\[
\Prob\Big\{\Big(R_k^{(m_1,\ldots,m_n)}\Big)^2<r^2\Big\}
\geq\prod_{a=1}^n\Big(\frac{r^{\frac{2}{n}}}{2(k+m_a)}\Big)^{k+m_a},
\]
and we thus obtain
\begin{align*}
\Prob\left\{{\mathcal{N}}_{GG}^{(m_1,\ldots,m_n)}(r;\infty)\geq q\right\}
&\geq\prod\limits_{k=1}^q\Prob\left\{\left(R_k^{(m_1,\ldots,m_n)}\right)^2<r^2\right\}
\geq\prod\limits_{k=1}^q\prod\limits_{a=1}^n\bigg(\frac{r^{\frac{2}{n}}}{2(k+m_a)}\bigg)^{k+m_a} \\
&=\exp\left\{\log\frac{r^{\frac{2}{n}}}{2}\sum\limits_{a=1}^n\sum\limits_{k=1}^q(k+m_a)
-\sum\limits_{a=1}^n\sum\limits_{k=1}^q(k+m_a)\log(k+m_a)\right\}.
\end{align*}
The next step is to evaluate the sums. We have
\[
\sum\limits_{a=1}^n\sum\limits_{k=1}^q(k+m_a)=\frac{n}{2}q^2+q\sum\limits_{a=1}^n\left(m_a+\tfrac{1}{2}\right),
\]
and
\[
\sum\limits_{a=1}^n\sum\limits_{k=1}^q(k+m_a)\log(k+m_a)=\frac{n}{2}q^2\log q-\frac{n}{4}q^2+q\log q\sum\limits_{a=1}^n\left(m_a+\tfrac{1}{2}\right)+O\left(\log q\right),
\]
as $q\rightarrow\infty$. Therefore,
\begin{multline*}
\Prob\left\{{\mathcal{N}}_{GG}^{(m_1,\ldots,m_n)}(r;\infty)\geq q\right\}
\geq\exp\Bigl\{
-\frac{n}{2}q^2\log q
+\left(\frac{n}{4}+\log r-\frac{n}{2}\log 2\right)q^2\\
-q\log q\sum\limits_{a=1}^n\left(m_a+\tfrac{1}{2}\right)
+q\sum\limits_{a=1}^n\left(m_a+\tfrac{1}{2}\right)+O\left(\log q\right)\Bigr\},
\end{multline*}
as $q\rightarrow\infty$.
\qed

\subsection{Proof of Theorem \ref{TheoremOvercrowdingBounds} B) Upper bound for the overcrowding probability.}

In order to obtain an upper bound for $\Prob\left\{{\mathcal{N}}_{GG}^{(m_1,\ldots,m_n)}(r;\infty)\geq q\right\}$ we use again the Markov inequality (Proposition \ref{PropositionMarkovInequality}), this times with $A=\{0,r^2\}$, $\varphi(x)=x^{-b}$, $b\geq 0$. This gives
\begin{align*}
\Prob\left\{(R_k^{(m_1,\ldots,m_n)})^2<r^2\right\}&\leq
(r^2)^{b}\ \frac{\displaystyle{\int_{0}^{\infty}}t^{-b}G_{0,n}^{n,0}(t|k+m_1-1,\ldots,k+m_n-1)dt}{\prod_{a=1}^n\Gamma(k+m_a)}\\
&= (r^2)^{b} \prod_{a=1}^n\frac{\Gamma(k+m_a-b)}{\Gamma(k+m_a)}.
\end{align*}
By the inequality (\ref{GammaInequality}) for the Gamma function
\begin{multline*}
\Prob\left\{\left(R_k^{(m_1,\ldots,m_n)}\right)^2<r^2\right\}
\leq \exp\biggl\{nb+b\log r^2+\sum\limits_{a=1}^n\left(k+m_a-\tfrac{1}{2}-b\right)\log(k+m_a-b) \\
-\sum\limits_{a=1}^n\left(k+m_a-\tfrac{1}{2}\right)\log(k+m_a)
+\frac{1}{12}\sum\limits_{a=1}^n\frac{1}{k+m_a-b}\biggr\}.
\end{multline*}
When choosing $b=k-\frac{1}{2}$, we obtain
\begin{multline*}
\Prob\left\{\left(R_k^{(m_1,\ldots,m_n)}\right)^2<r^2\right\}
\leq \exp\biggl\{\left(k-\tfrac{1}{2}\right)\log r^2+n\left(k-\tfrac{1}{2}\right)+\sum\limits_{a=1}^nm_a\log\left(m_a+\tfrac{1}{2}\right)\\
-\sum\limits_{a=1}^n\left(k+m_a-\tfrac{1}{2}\right)\log(k+m_a)
+\frac{1}{12}\sum\limits_{a=1}^n\frac{1}{m_a+\frac{1}{2}}\biggr\}.
\end{multline*}
Furthermore, using  simple probabilistic arguments we have
\begin{multline}\label{ProbabilisticInequality}
\Prob\Big\{{\mathcal{N}}_{GG}^{(m_1,\ldots,m_n)}(r;\infty)\geq q\Big\}\leq  \\
\Prob\Big\{\sum_{k=1}^{q^2}\mathbb{I}\Big(\Big(R_k^{(m_1,\ldots,m_n)}\Big)^2<r^2\Big)\geq q\Big\}
+\sum_{k=q^2+1}^{\infty}\Prob\Big\{\Big(R_k^{(m_1,\ldots,m_n)}\Big)^2<r^2\Big\}.
\end{multline}
Here $\mathbb{I}(\cdot)$ denotes the characteristic function of a set.  Then the second term on the right hand side of the inequality above can be accordingly estimated
\[
\Prob\left\{\left(R_k^{(m_1,\ldots,m_n)}\right)^2<r^2\right\}\leq\exp\left\{-nk\log k(1+o(1))\right\},
\]
as $k\rightarrow\infty$. Therefore,
\[
\sum_{k=q^2+1}^{\infty}\Prob\left\{\left(R_k^{(m_1,\ldots,m_n)}\right)^2<r^2\right\}\leq \exp\left\{-nq^2\log q^2(1+o(1))\right\},
\]
as $q\rightarrow\infty$.
Finally let us estimate the first term on the right hand side of the inequality (\ref{ProbabilisticInequality}).
We obtain
\[
\Prob\bigg\{\sum\limits_{k=1}^{q^2}\mathbb{I}\left(\left(R_k^{(m_1,\ldots,m_n)}\right)^2<r^2\right)\geq q\bigg\}
\leq \binom{q^2}{q} \prod\limits_{k=1}^q\Prob\left\{\left(R_k^{(m_1,\ldots,m_n)}\right)^2<r^2\right\}.
\]
Using $\binom{q^2}{q}<q^{2q}$ we thus obtain
\[
\Prob\left\{{\mathcal{N}}_{GG}^{(m_1,\ldots,m_n)}(r;\infty)\geq q\right\} \leq  
\exp\Bigl\{-\frac{n}{2}q^2\log q+q^2\Big(\frac{3n}{4}+\log r\Big)+q\log q\Big(2-\sum\limits_{a=1}^{n}m_a\Big)+O(q)\Bigr\},
\]
as $q\rightarrow\infty$. The main statement of the theorem then follows from the bounds A) and B).
\qed


\section{Joint density for the generalised quaternion Ginibre ensemble}
\label{sec:quaternions}

In this Section we turn to the case of random matrices taken from the induced quaternion Ginibre ensemble. Quaternion matrices play a prominent role in the physical sciences; nonetheless, the literature on quaternion matrices is very fragmented and often underappreciated. For this reason, we find it appropriate to start this section with a summary of some known properties of quaternions and matrices of quaternions. In particular, we are interested in decomposition theorems for quaternion matrices. Note that the generalised Schur decomposition, stated in theorem~\ref{TheoremGeneralizedSchurDecompositionQuaternionMatrices}, is an essential tool for the discussion of the complex eigenvalues of products of quaternion matrices, but a rigorous proof of this result has \emph{not} previously appeared in the literature. For a review of some known results for matrices with quaternion entries we refer to~\cite{Mehta,Zhang:1997} and references within.

\subsection{Quaternions}

Let $\textbf{1}$, $\textbf{i}$, $\textbf{j}$, and $\textbf{k}$ be abstract generators with the following multiplication rules
\begin{gather*}
\textbf{1}^2=\textbf{1},
\qquad
\textbf{i}^2=\textbf{j}^2=\textbf{k}^2=\textbf{ijk}=-\textbf{1},   \\
\textbf{i}\textbf{j}=-\textbf{j}\textbf{i}=\textbf{k},\qquad
\textbf{j}\textbf{k}=-\textbf{k}\textbf{j}=\textbf{i},\qquad
\textbf{k}\textbf{i}=-\textbf{i}\textbf{k}=\textbf{j}.
\end{gather*}
The division ring $\Hq$ of quaternions is  a four dimensional vector space over $\R$ with basis $\textbf{1}$, $\textbf{i}$, $\textbf{j}$ and  $\textbf{k}$.

\subsection{Representation of quaternions in terms of pairs of complex numbers}

We identify  quaternions of the form $x_0\textbf{1}$ with the real numbers, and we write $x_0$ instead of
$x_0\textbf{1}$. Moreover, we identify a quaternion of the form $x_0\textbf{1}+x_1\textbf{i}$ with the complex number
$x_0+ix_1$. This includes the real numbers $\R$, and the complex numbers $\C$ into $\Hq$ in the obvious way. Now, each quaternion $\textbf{x}=x_0\textbf{1}+x_1\textbf{i}+x_2\textbf{j}+x_3\textbf{k}$ can be written as
\[
\textbf{x}=(x_0+x_1\textbf{i})+(x_2+x_3\textbf{i})\textbf{j}.
\]
Therefore, each quaternion $\textbf{x}$ can be represented as
\[
\textbf{x}=\zeta+\eta\textbf{j},
\]
where $\zeta=x_0+x_1i$ and $\eta=x_2+x_3i$ are two complex numbers.

If $\textbf{x}=\zeta_1+\eta_1\textbf{j}$, and $\textbf{y}=\zeta_2+\eta_2\textbf{j}$ are two quaternions, then we find that the quaternion $\textbf{x}\textbf{y}$ can be written as
\[
\textbf{x}\textbf{y}=\zeta_1\zeta_2-\eta_1\overline{\eta_2}+(\eta_1\overline{\zeta_2}+\zeta_1\eta_2)\textbf{j}.
\]
Therefore the quaternions $\textbf{x}=\zeta_1+\eta_1\textbf{j}$, and $\textbf{y}=\zeta_2+\eta_2\textbf{j}$ can be multiplied in a formal way taking into account that
\[
\textbf{j}\alpha=\overline{\alpha}\textbf{j},\qquad \alpha\in\C.
\]

\subsection{Representation of quaternions as matrices with complex entries}
\label{SectionQuaternionsAsMatrices}

The division ring $\Hq$ of quaternions is isomorphic to the ring of $2\times 2$ matrices with complex entries
\begin{equation}\label{uvmatrix}
\Hq'=\left\{\left(\begin{array}{cc}
                    \alpha & -\overline{\beta} \\
                    \beta & \overline{\alpha}
                  \end{array}
\right)\biggr|\alpha,\beta\in\C\right\}.
\end{equation}
The ring $\Hq'$ is the sub-ring of $M_2(\C)$ under the operations of $M_2(\C)$. The isomorphism between $\Hq$ and $\Hq'$
is defined by
\[
\textbf{x}=\zeta+\eta\textbf{j}\rightsquigarrow \breve{x}=\left(\begin{array}{cc}
                                                                  \zeta & -\overline{\eta} \\
                                                                  \eta & \bar{\zeta}
                                                                \end{array}
\right).
\]
It can be checked that the bijection above preserves ring operations.

\subsection{The conjugate of a quaternion and the norm of a quaternion}

If $\textbf{x}\in\Hq$, and $\textbf{x}=x_0\textbf{1}+x_1\textbf{i}+x_2\textbf{j}+x_3\textbf{k}$,
then define the conjugate by
\[
\overline{\textbf{x}}=x_0\textbf{1}-x_1\textbf{i}-x_2\textbf{j}-x_3\textbf{k}.
\]
The definition of $\overline{\textbf{x}}$ implies that if $\textbf{x}=\zeta+\eta\textbf{j}$, then
$\overline{\textbf{x}}=\overline{\zeta}-\eta\textbf{j}$. The operation
$\textbf{x}\rightarrow \overline{\textbf{x}}$ in $\Hq$ corresponds to Hermitian conjugation in $\Hq'$. Namely,
if $\textbf{x}=\zeta+\eta\textbf{j}\in\Hq$ and $\breve{x}$ is the corresponding element in $\Hq'$, then the map $\textbf{x}\rightarrow\overline{\textbf{x}}$ in $\Hq$
induces the map
\[
\breve{x}=\left(\begin{array}{cc}
                                                                  \zeta & -\overline{\eta} \\
                                                                  \eta & \bar{\zeta}
                                                                \end{array}
\right)\longrightarrow\breve{x}^{*}=\left(\begin{array}{cc}
                                                                  \overline{\zeta} & \overline{\eta} \\
                                                                  -\eta & \zeta
                                                                \end{array}
\right)
\]
in $\Hq'$.
If $\textbf{x}\in\Hq$, and $\textbf{x}=x_0\textbf{1}+x_1\textbf{i}+x_2\textbf{j}+x_3\textbf{k}$, then the norm of $\textbf{x}$ is defined by
\[
|\textbf{x}|=\sqrt{\overline{\textbf{x}}\textbf{x}}=\sqrt{x_0^2+x_1^2+x_2^2+x_3^2}.
\]
If $|\textbf{x}|=1$, then $\textbf{x}$ is called a unit quaternion. If $\textbf{x}$ is a unit quaternion, and $\textbf{x}=\zeta+\eta\textbf{j}$,
then $|\zeta|^2+|\eta|^2=1$. The corresponding element $\breve{x}$ from $\Hq'$ is a $2\times 2$ unitary matrix.

\subsection{Quaternion matrices}
\label{SectionQuaternionmatrices}

Let $\mbox{Mat}\left(\Hq,m\times n\right)$ denote the collection of all $m\times n$ matrices with entries from $\Hq$. Since for $\textbf{x}\in\Hq$ there is a unique representation $\textbf{x}=\zeta+\eta\textbf{j}$ (where $\zeta,\eta\in\C$), we conclude that for $\textbf{P}\in \mbox{Mat}\left(\Hq, m\times n\right)$ there is a unique representation as $\textbf{P}=A+B\textbf{j}$, where $A,B\in \mbox{Mat}\left(\C,m\times n\right)$. The quaternion matrix $\textbf{P}=A+B\textbf{j}$ can be also represented as the $2m\times 2n$ complex matrix
\begin{equation}\label{Form}
\hat{P}=\left(\begin{array}{cc}
                A & -\overline{B} \\
                B & A
              \end{array}
\right).
\end{equation}
The map $\textbf{P}\rightarrow \hat{P}$ is an algebra isomorphism of $\mbox{Mat}\left(\Hq,m\times n\right)$ onto the algebra of $2m\times 2n$ complex matrices of  form
(\ref{Form}).

Under the isomorphism $\Hq\rightarrow\Hq'$ defined in Section \ref{SectionQuaternionsAsMatrices} the set $\mbox{Mat}\left(\Hq,m\times n\right)$ turns into
$\mbox{Mat}\left(\Hq',m\times n\right)$, i.e. to the collection of all $m\times n$ matrices whose entries are $2\times 2$ blocks of the form
\[
\left(\begin{array}{cc}
        u & -\overline{v} \\
        v & \overline{u}
      \end{array}
\right),\qquad u,v\in\C.
\]
Thus each matrix $M\in\mbox{Mat}\left(\Hq',m\times n\right)$ can be represented as
\begin{equation}
M=\left(\breve{m}_{i,j}\right)_{i,j=1}^N,\;\; \breve{m}_{i,j}=\left(\begin{array}{cc}
                                                                      M_{i,j}^{(1)} & -\overline{M}_{i,j}^{(2)} \\
                                                                     M_{i,j}^{(2)} &  \overline{M}_{i,j}^{(1)}
                                                                    \end{array}
\right),
\end{equation}
where $1\leq i\leq m$, $1\leq j\leq n$, and where
$ M_{i,j}^{(1)},  M_{i,j}^{(2)}\in\C$.

If $\textbf{P}=(\,\textbf{p}_{ij}\,)_{i,j=1}^n\in\mbox{Mat}\left(\Hq,n\times n\right)$, i.e. a $n\times n$ matrix with entries $\textbf{p}_{i,j}\in\Hq$, then
\[
\textbf{P}^{*}=(\,\overline{\textbf{p}_{ji}}\,)_{i,j=1}^n
\qquad \text{and} \qquad
\overline{\textbf{P}}=(\,\overline{\textbf{p}_{ij}}\,)_{i,j=1}^n.
\]

If $\textbf{P}^*\textbf{P}=\textbf{P}\textbf{P}^*=\textbf{I}$, where $\textbf{I}\equiv\diag(\textbf{1},\ldots,\textbf{1})$ is the identity matrix,
then we say that $\textbf{P}$ is a quaternion unitary matrix. Note that if $\textbf{P}$ is a quaternion unitary matrix, then the corresponding element from
$\mbox{Mat}\left(\Hq',n\times n\right)$ is an $2\times 2$ block unitary matrix with complex entries, which is a unitary symplectic matrix.

A matrix $\textbf{P}\in\Mat\left(\Hq,n\times n\right)$ is non-singular if there exists a matrix $\textbf{Q}\in\Mat\left(\Hq,n\times n\right)$
such that $\textbf{P}\textbf{Q}=\textbf{Q}\textbf{P}=\textbf{I}$.

\subsection{Eigenvalues of quaternion matrices}
\label{SectionAppendixEigenvaluesQuaternionmatrices}

Denote by $\C^{+}$ the set of complex numbers with non-negative imaginary parts. Recall that the real and complex numbers were embedded
into $\Hq$. Thus we have
\[
\R\subset\C^{+}\subset\C\subset\Hq.
\]

\begin{defn}
A number $\lambda\in\C^{+}$ is called an eigenvalue of $\textbf{P}\in\Mat\left(\Hq,n\times n\right)$ if there exists a non-zero vector $\textbf{x}\in\Mat\left(\Hq,n\times 1\right)$
such that $\textbf{P}\textbf{x}=\textbf{x}\lambda$.
\end{defn}

\begin{rem}
If $\lambda\in\C^+$ is an eigenvalue of $\textbf{P}\in\Mat\left(\Hq,n\times n\right)$ with an eigenvector $\textbf{x}\in\Mat\left(\Hq,n\times 1\right)$,
then
\[
\textbf{P}(\textbf{x}\textbf{w})=(\textbf{x}\textbf{w})(\overline{\textbf{w}}\lambda\textbf{w})
\]
for any $\textbf{w}\in\Hq$ with $|\textbf{w}|=1$. Therefore the set
$\{\overline{\textbf{w}}\lambda\textbf{w}\,|\,|\textbf{w}|=1\}$ can be understood as a continuum of ``eigenvalues''
each of which is a representative of $\lambda$.
\end{rem}

\begin{prop}\label{PropositionEigenvaluesQuaternionMatrices}
Any quaternion matrix $\textbf{P}\in\Mat\left(\Hq,n\times n\right)$ has exactly
$n$ eigenvalues.
\end{prop}
\begin{proof}
Assume that $\lambda$ is an eigenvalue of $\textbf{P}$. Then $\textbf{P}\textbf{x}=\textbf{x}\lambda$ for some non-zero vector $\textbf{x}$,
$\textbf{x}\in\Mat\left(\Hq,n\times 1\right)$. We write
\[
\textbf{P}=P_1+P_2\textbf{j},\qquad \textbf{x}=x_1+x_2\textbf{j},
\]
where $P_1$, $P_2$ are complex matrices, and where $x_1$, $x_2$ are complex column vectors.
Since $\lambda$ is a complex number, the equation $\textbf{P}\textbf{x}=\textbf{x}\lambda$
is equivalent to
\[
\begin{pmatrix}
P_1 & P_2 \\
-\overline{P_2} & \overline{P_1}
\end{pmatrix}
\begin{pmatrix}
x_1 & x_2 \\
-\overline{x_2} & \overline{x_1}
\end{pmatrix}
=
\begin{pmatrix}
x_1 & x_2 \\
-\overline{x_2} & \overline{x_1}
\end{pmatrix}
\begin{pmatrix}
\lambda & 0 \\
0 & \overline{\lambda}
\end{pmatrix},
\]
where we have used the isomorphism defined by equation~\eqref{Form}, i.e. we are using the representation $\hat P\hat x=\hat x\hat\lambda$. We obtain two equations,
\[
\begin{pmatrix} P_1 & P_2 \\ -\overline{P_2} & \overline{P_1} \end{pmatrix}
\begin{pmatrix} x_1 \\ -\overline{x_2} \end{pmatrix}=
\lambda
\begin{pmatrix} x_1 \\ -\overline{x_2} \end{pmatrix}
\quad\text{and}\quad
\begin{pmatrix} P_1 & P_2 \\ -\overline{P_2} & \overline{P_1} \end{pmatrix}
\begin{pmatrix} x_2 \\ -\overline{x_1} \end{pmatrix}=
\overline\lambda
\begin{pmatrix} x_2 \\ -\overline{x_1} \end{pmatrix},
\]
which are equivalent to each other.
Since $\hat P$ is a $2n\times 2n$ complex matrix, it has exactly $2n$ complex eigenvalues. Now, in order
to complete the proof of the Proposition it is enough to show that the non-real eigenvalues of
$\hat P$
come in conjugate pairs, and that the real eigenvalues occur an even number of times.

Note that
\[
\left(\begin{array}{cc}
        0 & 1 \\
        -1 & 0
      \end{array}
\right)\left(\begin{array}{cc}
        P_1 & P_2 \\
        -\overline{P_2} & \overline{P_1}
      \end{array}
\right)\left(\begin{array}{cc}
        0 & -1 \\
        1 & 0
      \end{array}
\right)=\left(\begin{array}{cc}
        \overline{P_1} & \overline{P_2} \\
        -P_2 & P_1
      \end{array}
\right),
\]
and
\[
\det\left(\lambda I_{2n}-\left(\begin{array}{cc}
        P_1 & P_2 \\
        -\overline{P_2} & \overline{P_1}
      \end{array}
\right)\right)=\overline{\det\left(\overline{\lambda} I_{2n}-\left(\begin{array}{cc}
        P_1 & P_2 \\
        -\overline{P_2} & \overline{P_1}
      \end{array}
\right)\right)}.
\]
This says that if $\lambda$ is a non-real eigenvalue, then $\overline{\lambda}$ is an eigenvalue as well. By a simple continuity
argument we conclude that real eigenvalues appear an even number of times.
\end{proof}

\begin{rem}
Proposition \ref{PropositionEigenvaluesQuaternionMatrices} implies that eigenvalues of quaternion matrices
(defined as elements of $\C^+$) form a discrete, finite set.
\end{rem}

\begin{cor}\label{CorrolaryEigenvaluesQuaternionMatrices}
Assume that $\textbf{P}\in\Mat\left(\Hq,n\times n\right)$, and $P'$ is the corresponding element of $\Mat\left(\Hq',n\times n\right)$
obtained from $\textbf{P}$ by the isomorphism defined in Section \ref{SectionQuaternionsAsMatrices}. Then $P'$ has $2n$ eigenvalues, where the non-real eigenvalues of $P'$ appear in conjugate pairs, and every real eigenvalue of $P'$ occurs an even number of times.
\end{cor}

\begin{proof}
Let $P'\in\Mat\left(\Hq',n\times n\right)$. We know that the total number of eigenvalues of $P'$ is $2n$. Let $\textbf{P}$ be the pre-image
of $P'$ under the isomorphism between $\Mat\left(\Hq,n\times n\right)$ and $\Mat\left(\Hq',n\times n\right)$. It can be checked directly
that if $\lambda$ is an eigenvalue of $\textbf{P}$, then both $\lambda$ and $\bar{\lambda}$ are eigenvalues of $P'$. This observation, and the fact that
$\textbf{P}$  has exactly $n$ eigenvalues (see Proposition \ref{PropositionEigenvaluesQuaternionMatrices}) imply the statement
of  Corollary \ref{CorrolaryEigenvaluesQuaternionMatrices}.
\end{proof}

\subsection{The Schur canonical form for quaternion matrices}

\begin{lem}
\label{Lemma1Quaternion}
Assume that $\textbf{P}\in\Mat\left(\Hq,m\times n\right)$, and assume that $m<n$. Then $\textbf{P}\textbf{x}=0$ has a non-trivial solution in $\Mat\left(\Hq,n\times 1\right)$.
\end{lem}

\begin{proof}
Represent $\textbf{P}$ as $\textbf{P}=A+B\textbf{j}$, and $\textbf{x}$ as $\textbf{x}=x_1+x_2\textbf{j}$.
Then $\textbf{P}\textbf{x}=0$ is equivalent to $\hat{P}\hat{x}=0$ (where $\hat{P}$ and $\hat{x}$ are images of $\textbf{P}$
and $\textbf{x}$ under the isomorphism defined in Section \ref{SectionQuaternionmatrices}). More explicitly,
$\textbf{P}\textbf{x}=0$ is equivalent to the equation
\[
\left(\begin{array}{cc}
        A & -\overline{B} \\
        B & \overline{A}
      \end{array}
\right)\left(\begin{array}{cc}
               x_1 & x_2 \\
               -\overline{x_2} & \overline{x_1}
             \end{array}
\right)=0
\qquad\text{or equivalently}\qquad
\left(\begin{array}{cc}
        A & -\overline{B} \\
        B & \overline{A}
      \end{array}
\right)\left(
               \begin{array}{c}
                 x_1 \\
                 -\overline{x_2} \\
               \end{array}
             \right)
=0.
\]
The system just written above does have a non-trivial solution.
\end{proof}

\begin{prop}\label{propositionQuaternionUnitaryColumn}
Assume that $\textbf{u}_1\in\Mat(\Hq,n\times 1)$, and that $\textbf{u}_1^*\textbf{u}_1=1$.
Then there exists a unitary matrix from $\Mat(\Hq,n\times n)$ whose first column is $\textbf{u}_1$.
\end{prop}

\begin{proof}
According to Lemma \ref{Lemma1Quaternion} there exists a non-trivial vector $\tilde{\textbf{u}}_2$, $\tilde{\textbf{u}}_2\in\Mat(\Hq,n\times 1)$, such that
$\textbf{u}_1^{*}\tilde{\textbf{u}}_2=0$. Define $\textbf{u}_2=\frac{\tilde{\textbf{u}}_2}{|\tilde{\textbf{u}}_2|}$. By the same argument there exists a non-trivial vector
 $\tilde{\textbf{u}}_3$, $\tilde{\textbf{u}}_3\in\Mat(\Hq,n\times 1)$, such that
$\textbf{u}_2^{*}\tilde{\textbf{u}}_3=0$. Define $\textbf{u}_3=\frac{\tilde{\textbf{u}}_3}{|\tilde{\textbf{u}}_3|}$. Proceeding in this way we obtain $n$ vectors $\textbf{u}_1$, $\textbf{u}_2$, $\ldots$, $\textbf{u}_n$.
The desired unitary matrix is $\textbf{U}=\left(\textbf{u}_1,\textbf{u}_2,\ldots,\textbf{u}_n\right)$.
\end{proof}

\begin{prop}\label{TheoremSchurCanonicalForm}
If $\textbf{P}\in\Mat(\Hq,n\times n)$, then there exists a unitary quaternion matrix $\textbf{U}\in\Mat(\Hq,n\times n)$
such that $\textbf{U}^*\textbf{P}\textbf{U}$ is in an upper triangular form.
\end{prop}

\begin{proof}
Proposition \ref{PropositionEigenvaluesQuaternionMatrices} implies that $\textbf{P}$ has exactly $n$   eigenvalues $\lambda_1,\ldots,\lambda_n$.
Let $\textbf{x}_1$ be a normalised eigenvector of $\textbf{P}$ corresponding to $\lambda_1$. By Proposition \ref{propositionQuaternionUnitaryColumn} there exists a unitary matrix $\textbf{U}_1\in\Mat(\Hq,n\times n)$ whose first column is $\textbf{x}_1$. Then the matrix $\textbf{U}_1^{*}\textbf{P}\textbf{U}_1$ has the form
\[
\textbf{U}_1^{*}\textbf{P}\textbf{U}_1=\left(\begin{array}{cc}
                                               \lambda_1 & * \\
                                               0 & \textbf{P}_1
                                             \end{array}
\right).
\]
The matrix $\textbf{P}_1\in\Mat(\Hq,(n-1)\times (n-1))$ has $n-1$ eigenvalues $\lambda_2,\ldots,\lambda_n$. Then we can find a unitary matrix
$\textbf{U}_2\in\Mat(\Hq,(n-1)\times (n-1))$ such that
\[
\textbf{U}_2^{*}\textbf{P}_1\textbf{U}_2=\left(\begin{array}{cc}
                                               \lambda_2 & * \\
                                               0 & \textbf{P}_2
                                             \end{array}
\right).
\]
If we set $\textbf{V}_2=\diag(\textbf1,\textbf U_2)$, then
\[
(\textbf{U}_1\textbf{V}_2)^*P\textbf{U}_1\textbf{V}_2=
\left(\begin{array}{ccc}
\lambda_1 & * & *  \\
0 & \lambda_2 & *  \\
0 & 0 & \textbf{P}_2
\end{array}
\right),
\]
where $\textbf{U}_1\textbf{V}_2$ is a unitary matrix from $\Mat\left(\Hq,n\times n\right)$.
Repeating this argument we obtain the statement of Proposition \ref{TheoremSchurCanonicalForm}.
\end{proof}

\begin{cor}\label{CorollaryUZPlusTUQuaternion}
If $M\in\Mat\left(\Hq',n\times n\right)$, then there exists a unitary $2\times 2$ block matrix $U\in\Mat\left(\Hq',n\times n\right)$, that is a unitary symplectic matrix $U\in\USp(2n)$, such that
\begin{equation}\label{MZTU}
M=U(Z+T)U^*.
\end{equation}
Here $Z\in\Mat\left(\Hq',n\times n\right)$ is a $2\times 2$ block diagonal matrix of the form
\begin{equation}\label{Zquaternion}
Z=\left(\begin{array}{cccc}
          \breve{z}_1 & \breve{0} & \cdots & \breve{0} \\
          \breve{0} & \breve{z}_2 & \cdots & \breve{0} \\
          \vdots & \vdots & \ddots & \vdots \\
          \breve{0} & \breve{0} & \cdots & \breve{z}_n
        \end{array}
\right),\qquad \breve{z}_i=\left(\begin{array}{cc}
                        z_i & 0 \\
                        0 & \overline{z_i}
                      \end{array}
\right),\qquad \breve{0}=\left(\begin{array}{cc}
                       0 & 0 \\
                        0 & 0
                      \end{array}
\right),
\end{equation}
(with $z_i\in\C$), and $T\in\Mat\left(\Hq',n\times n\right)$ is a $2\times 2$ block strictly upper triangular matrix of the form
\begin{equation}\label{Tquaternion}
T=\left(\begin{array}{ccccc}
          \breve{0} & \breve{t}_{1,2} & \breve{t}_{1,3} & \cdots & \breve{t}_{1,n} \\
          \breve{0} & \breve{0} & \breve{t}_{2,3} & \cdots  & \breve{t}_{2,n} \\
          \vdots & \vdots & \vdots & \ddots & \vdots \\
           \breve{0} & \breve{0} & \breve{0} & \cdots  & \breve{t}_{n-1,n} \\
          \breve{0} & \breve{0} & \breve{0} & \cdots & \breve{0}
        \end{array}
\right),
\qquad
\breve{t}_{i,j}=\left(\begin{array}{cc}
                                    t_{i,j}^{(1)} & -\overline{t_{i,j}^{(2)}} \\
                                    t_{i,j}^{(2)} & \overline{t_{i,j}^{(1)}}
                                  \end{array}
\right),
\end{equation}
with $t_{i,j}^{(1)},t_{i,j}^{(2)}\in\C$.
\end{cor}

\begin{proof}
Use the Schur canonical form for quaternion matrices (see Proposition \ref{TheoremSchurCanonicalForm}),
and the isomorphism between  $\Mat\left(\Hq,n\times n\right)$ and $\Mat\left(\Hq',n\times n\right)$
defined in Section \ref{SectionQuaternionmatrices}.
\end{proof}

\begin{rem} Similar to the standard Schur decomposition, decomposition (\ref{MZTU}) is not unique.
However, in subsequent calculations leading to Jacobian determinantal formulas we omit all matrices with equal eigenvalues (such matrices are of zero
Lebesgue measure).
Then the uniqueness of decomposition (\ref{MZTU}) can be restored by requiring that $z_i$ are in increasing order (with respect to the lexicographic order on complex numbers,
i.e. $u+iv\leq u'+iv'$ if $u<u'$ or if $u=u'$ and $v\leq v'$), and then by requiring that $z_i\in\C^+$.
\end{rem}

\subsection{QR decomposition for quaternion matrices}

\begin{prop}
\label{TheoremQRFactorizationQuaternionMatrices}
Let $\textbf{P}\in\Mat\left(\Hq,n\times n\right)$ be a non-singular quaternion matrix.
Then there is a factorization
\[
\textbf{P}=\textbf{U}\textbf{S},
\]
where $\textbf{U}\in\Mat\left(\Hq,n\times n\right)$ is unitary, and $\textbf{S}\in\Mat\left(\Hq,n\times n\right)$
is an upper triangular.
\end{prop}

\begin{proof}
Similar to the  case of matrices with complex entries, this statement is just a reformulation (in matrix notation)
of the result of applying of the Gram--Schmidt process to the columns of $\textbf{P}$.
\end{proof}

\begin{cor}
Let $P\in\Mat\left(\Hq',n\times n\right)$ be a non-singular matrix. Then there is a factorization
\[
P=US,
\]
where $U$ is a unitary $2\times 2$ block matrix, $U\in\Mat\left(\Hq',n\times n\right)$,
and $S\in\Mat\left(\Hq',n\times n\right)$ is a $2\times 2$ block upper triangular matrix.
\end{cor}

\begin{proof}
Use the QR factorization for quaternion matrices (Proposition \ref{TheoremQRFactorizationQuaternionMatrices}), and the isomorphism between  $\Mat\left(\Hq,n\times n\right)$ and $\Mat\left(\Hq',n\times n\right)$
defined in Section \ref{SectionQuaternionmatrices}.
\end{proof}

\subsection{Generalised Schur decomposition for quaternion matrices}
\label{SectionGeneralizedSchurDecompositionQuaternionMatrices}

\begin{thm}
\label{TheoremGeneralizedSchurDecompositionQuaternionMatrices}
Let $N$ and $n$ be fixed natural numbers. Let $\textbf{M}_i\ (i=1,\ldots,n)$
be non-singular quaternion matrices from  $\Mat\left(\Hq,N\times N\right)$. Then there exist unitary quaternion matrices
$\textbf{U}_i$ $(i=1,\ldots,n)$ from  $\Mat\left(\Hq,N\times N\right)$, and upper triangular quaternion matrices
 $\textbf{S}_i$ $(i=1,\ldots,n)$ from  $\Mat\left(\Hq,N\times N\right)$ such that
\begin{equation}\label{genSchur}
 \textbf{M}_i=\textbf{U}_i\textbf{S}_i\textbf{U}_{i+1}^*,\qquad i=1,\ldots,n,
\end{equation}
where $\textbf U_{n+1}= \textbf U_1$.
\end{thm}

\begin{proof}
Let $\textbf{M}_i\ (i=1,\ldots,n)$ be non-singular quaternion matrices from  $\Mat\left(\Hq,N\times N\right)$. It follows from the Schur decomposition, Proposition~\ref{TheoremSchurCanonicalForm}, that there exists a unitary quaternion matrix $\textbf U_{1}\in U(N,\Hq)$, such that
\begin{equation}\label{prod-proof}
\textbf U_{1}^*\textbf M_1\textbf M_2\cdots\textbf M_n\textbf U_{1}=\textbf S,
\end{equation}
where $\textbf{S}$ is an upper triangular quaternion matrix from  $\Mat\left(\Hq,N\times N\right)$.
It is clear that $\textbf{M}_n\textbf U_{1}\in\Mat\left(\Hq,N\times N\right)$ and it follows directly from the QR decomposition for quaternion matrices, Proposition~\ref{TheoremQRFactorizationQuaternionMatrices}, that there exists a unitary quaternion matrix $\textbf U_{n}\in U(N,\Hq)$, such that
\[
 \textbf{M}_n\textbf{U}_{1}=\textbf{U}_n\textbf{S}_n,
\]
where $\textbf{S}_n$ is an upper triangular quaternion matrix from  $\Mat\left(\Hq,N\times N\right)$. This expression is identical to~\eqref{genSchur} with $i=n$. Repeating this procedure, we also find 
\begin{equation*}
 \textbf{M}_i\textbf{U}_{i+1}=\textbf{U}_i\textbf{S}_i,\qquad i=2,\ldots,n-1,
\end{equation*}
with $\textbf{S}_i\ (i=2,\ldots,n-1)$  upper triangular quaternion matrices from  $\Mat\left(\Hq,N\times N\right)$ and $\textbf U_{i}\ (i=2,\ldots,n)$ unitary quaternion matrices. Using this in~\eqref{prod-proof} we get
\begin{equation*}
\textbf U_{1}^*\textbf M_1\textbf U_{2}\textbf S_2\cdots\textbf S_n=\textbf S.
\end{equation*}
The upper triangular matrices $\textbf{S}_i\ (i=2,\ldots,n)$ are invertible, since the matrices $\textbf{M}_i\ (i=2,\ldots,n)$ are non-singular. We define the upper triangular matrix $\textbf{S}_1\equiv \textbf{S}\textbf S_2^{-1}\cdots\textbf S_n^{-1}$ and the theorem follows.
\end{proof}

\begin{rem}
Note that the diagonal elements of the upper triangular matrix $\q S$ are complex numbers, i.e. $(\q S)_{ii}\in\C\subset\Hq$ for all $i=1,\ldots,N$. This is equivalent to the structure given by Corollary~\ref{CorollaryUZPlusTUQuaternion}.
\end{rem}

\begin{rem}
Similar to the ordinary Schur decomposition, the generalised Schur decomposition given by Theorem~\ref{TheoremGeneralizedSchurDecompositionQuaternionMatrices} is not unique. In particular, let $\q z_a=x+y\q i\in\C\subset\Hq$ with $|\q z_a|=1$ for all $a=1,\ldots,N$ and $\q V=\diag(\q z_1,\ldots,\q z_N)$, then with the same definitions as in Theorem~\ref{TheoremGeneralizedSchurDecompositionQuaternionMatrices}, we have
\[
 \textbf{M}_i=\textbf{U}_i\textbf{S}_i\textbf{U}_{i+1}^*=(\textbf{U}_i\q V)(\q V^*\textbf{S}_i)\textbf{U}_{i+1}^*,
\]
where $\textbf{U}_i\q V\in U(N,\Hq)$ is a unitary quaternion matrix and $\q V^*\textbf{S}_i$ is a upper triangular quaternion matrices from  $\Mat\left(\Hq,N\times N\right)$. Uniqueness of the generalised Schur decomposition (Theorem~\ref{TheoremGeneralizedSchurDecompositionQuaternionMatrices}) can be obtained by choosing $\q U_i\in U(N,\Hq)/U(1,\C)^N$ such that the diagonal elements $(\q S)_{jj}\in\C^+$ ($j=1,\ldots,N$) are complex numbers in lexicographical order.
\end{rem}

\begin{cor}
\label{CorrollaryGeneralizedSchurDecompositionQuaternionMatrices}
Let $N$ and $n$ be fixed natural numbers. Let $M_i\ (i=1,\ldots,n)$
be non-singular $2\times 2$ block matrices from  $\Mat\left(\Hq',N\times N\right)$. There exist unitary $2\times 2$ block matrices
$U_i\ (i=1,\ldots,n)$ from  $\Mat\left(\Hq',N\times N\right)$, and upper triangular $2\times 2$ block matrices
 $S_i\ (i=1,\ldots,n)$ from  $\Mat\left(\Hq',N\times N\right)$ such that
 \begin{equation}\label{1141}
 M_i=U_iS_iU_{i+1}^*,\qquad i=1,\ldots,n
 \end{equation}
where $U_{n+1}= U_1$.
\end{cor}

\begin{proof}
Use the generalised Schur decomposition for quaternion matrices (Theorem \ref{TheoremGeneralizedSchurDecompositionQuaternionMatrices}),
and the isomorphism between  $\Mat\left(\Hq,n\times n\right)$ and $\Mat\left(\Hq',n\times n\right)$
defined in Section \ref{SectionQuaternionmatrices}.
\end{proof}

\begin{rem}
\label{ZplusT}
Often we write the upper triangular matrices $S_i\ (i=1,\ldots,n)$ as a sum,
$S_i=Z_i+T_i$, consisting of a diagonal matrix $Z_i$ similar to~\eqref{Zquaternion}, and a strictly upper triangular matrix $T_i$ similar to~\eqref{Tquaternion}.
\end{rem}

\subsection{Change of measure for quaternion matrices}

Corollary \ref{CorrollaryGeneralizedSchurDecompositionQuaternionMatrices} defines a change of variables from the matrices $M_1, \ldots, M_n$ to their block triangular forms. Our task is to derive the Jacobian for such a change of variables. To make our method of derivation more transparent we begin our presentation from the simplest case of a single matrix.
\begin{defn}\label{DefinitionQuaternionOneForm}
Let $\lambda^{(1)}$, $\lambda^{(2)}$ be independent one-forms. The $2\times 2$ matrix one-form $\breve{\lambda}$ defined by
\begin{equation}\label{QuaternionOneForm}
\breve{\lambda}=\left(\begin{array}{cc}
                        \lambda^{(1)} & -\overline{\lambda}^{(2)} \\
                        \lambda^{(2)} & \overline{\lambda}^{(1)}
                      \end{array}
\right)
\end{equation}
is called a quaternion one-form.
\end{defn}
\begin{rem}
Recall that the set $\Hq'$, see~\eqref{uvmatrix},
 is isomorphic to the set $\Hq$ of quaternions. This
explains Definition~\ref{DefinitionQuaternionOneForm}.
\end{rem}
We introduce the following notation. If $\breve{\lambda}$ is a quaternion one-form defined by equation (\ref{QuaternionOneForm}),
then
\[
\underline{\breve{\lambda}}=\lambda^{(1)}\wedge\overline{\lambda}^{(1)}\wedge\lambda^{(2)}\wedge\overline{\lambda}^{(2)}.
\]
So $\underline{\breve{\lambda}}$ is the wedge product of one-forms which are entries of
$\breve{\lambda}$.

Let $M\in\Mat\left(\Hq',N\times N\right)$. Represent $M$  as in equation (\ref{M}).
Integration of a function of $M$ with respect to the Lebesgue measure is the same as integrating
against the following $4N^2$ form
\[
\underline{dM}\equiv\bigwedge_{i,j} \underline{d\breve{M}_{i,j}}.
\]
Here we have introduced $\underline{dM}$ as a compact notation for the $4N^2$ form.

\begin{prop}
\label{PropositionTQ}
Let $M$, $Q_1$, $Q_2$ be elements of $\Mat(\Hq',n\times n)$. If $Q_1$, $Q_2$ are fixed unitary matrices, and $\tilde{M}=Q_1MQ_2$, then $\underline{d\tilde M}=\underline{dM}$.
\end{prop}

\begin{proof}
Use the fact that if $x_i'=\sum_{j}a_{ij}x_j$, 
and $A=(a_{ij})_{i,j}$ is a constant matrix, then
\[
dx_1'\wedge\ldots \wedge dx_n'=\det(A)dx_1\wedge\ldots\wedge dx_n.
\]
\end{proof}

We use Corollary \ref{CorollaryUZPlusTUQuaternion}, and write $M$ in terms of a unitary matrix
$U\in \USp(2N)/U(1)^N$, a diagonal matrix $Z$ and a strictly upper triangular matrix $T$ (see equations \eqref{MZTU} to \eqref{Tquaternion}).
Set $dM=(d\breve{m}_{ij})_{i,j=1}^N$, thus $dM$ can be understood as a matrix whose elements, $d\breve{m}_{i,j}$, are quaternion one-forms.
We have
\[
dM=U\left(\Omega S-S\Omega+dS\right)U^*,
\]
where $\Omega=U^*dU$ is skew-Hermitian and $S=Z+T$. Denote by $\breve{w}_{i,j}$ $(1\leq i,j\leq N)$ the $2\times 2$ blocks which are matrix elements of $\Omega$. Equivalent to $dM$, the matrices $dU$, $dZ$ and $dT$ can be understood as matrices whose elements are quaternion one-forms. We will use the abbreviations
\[
\underline{dZ}=\underset{i}{\bigwedge}(dz_i\wedge d\bar{z}_i), \qquad 
\underline{dT}=\bigwedge_{i<j}\underline{d\breve{t}_{i,j}} \qquad\text{and}\qquad
\underline{\Omega}=\bigwedge_i(w_{i,i}^{(2)}\wedge \bar{w}_{i,i}^{(2)})\bigwedge_{i>j}\underline{\breve{w}_{i,j}};
\]
that is the wedge product of all independent one-forms. Furthermore, we use the abbreviation 
\begin{equation}\label{LambdaDef}
\Lambda=(\,\breve\lambda_{i,j}\,)_{i,j=1}^N=\Omega S-S\Omega+dS
\qquad\text{with}\qquad
\underline{\Lambda}=\bigwedge_{i,j} \underline{\breve{\lambda}_{i,j}}.
\end{equation}

\begin{thm}
\label{TheoremQuaternionJacobianDeterminantalFormula}
The following Jacobian determinantal formula holds true
\begin{equation}\label{JacobianDeterminantalFormulaQuaternion1}
\underline{dM}
=\prod\limits_{j<i}|z_j-z_i|^2|z_j-\bar{z}_i|^2\prod\limits_{j=1}^N|z_j-\overline{z}_j|^2
\underline{dZ} \wedge \underline{dT} \wedge \underline{\Omega}.
\end{equation}
\end{thm}

\begin{rem}
1) Equation (\ref{JacobianDeterminantalFormulaQuaternion1}) can be understood as an analogue of the complex Ginibre measure decomposition, see equation (6.3.5)
of Hough,  Krishnapur,  Peres, and Vir\'ag  \cite{Hough}.\\
2) Here (and in subsequent calculations leading to the equation in the statement of Theorem \ref{TheoremQuaternionJacobianDeterminantalFormula},
and in statement of Theorem \ref{TheoremProductQuaternionJacobiDeterminantalFormula}) we omit combinatorial constants before differential forms.
In what follows we will restore normalization constants for the relevant probability measures using the usual normalization condition.\\
3) Our proof of Theorem \ref{TheoremQuaternionJacobianDeterminantalFormula} can be seen as an extension of that in Section 6.3 of
Hough,  Krishnapur,  Peres, and Vir\'ag  \cite{Hough} to the case of matrices from $\Mat\left(\Hq',N\times N\right)$.
\end{rem}

\begin{proof}
Proposition \ref{PropositionTQ} implies $\underline{dM}=\underline{\Lambda}$.
The explicit formula for $\breve{\lambda}_{i,j}$ is
\begin{equation}\label{brevelambda}
\breve{\lambda}_{i,j}=\sum_{k=1}^{j}\breve{w}_{i,k}\breve{s}_{k,j}-\sum_{k=i}^N\breve{s}_{i,k}\breve{w}_{k,j}+d\breve s_{i,j},
\end{equation}
where $d\breve s_{i,j}=\breve 0$ for $i>j$. Recall that the quaternion one-forms $\breve{w}_{i,j}$ are the matrix elements of $\Omega$. We emphasise that $\breve{w}_{i,j}^*=\breve{w}_{j,i}$ for $(i<j)$ and that the diagonal elements $\breve{w}_{i,i}$ $(i=1,\ldots,N)$ have the special structure
\[
\breve{w}_{i,i}=\begin{pmatrix} 0 & -\overline w_{i,i}^{(2)} \\ w_{i,i}^{(2)} & 0 \end{pmatrix}.
\]
Now we begin to investigate how the terms $\underline{\breve{\lambda}_{i,j}}$
contribute to the wedge product $\underline{\Lambda}$, for different choices of indices $i$ and $ j$.
We start with $i=N$ and $j=1$, and find that
\[
\breve{\lambda}_{N,1}=\breve{w}_{N,1}\breve{s}_{1,1}-\breve{s}_{N,N}\breve{w}_{N,1}=\breve{w}_{N,1}\breve{z}_1-\breve{z}_N\breve{w}_{N,1}.
\]
This gives
\begin{equation}\label{N1}
\underline{\breve{\lambda}_{N,1}}=|z_1-z_N|^2|z_1-\overline{z}_N|^2\underline{\breve{w}_{N,1}}.
\end{equation}
For $\breve{\lambda}_{N,2}$ we find
\[
\breve{\lambda}_{N,2}=\breve{w}_{N,2}\breve{z}_2-\breve{z}_N\breve{w}_{N,2}+\breve{w}_{N,1}\breve{t}_{1,2}.
\]
Taking into account (\ref{N1}) it is not hard to see that the term $\breve{w}_{N,1}\breve{t}_{1,2}$ does not contribute to the wedge product
$\underline{\breve{\lambda}_{N,1}}\wedge\underline{\breve{\lambda}_{N,2}}$
(the one-forms consisting of $\breve{w}_{N,1}$ already have appeared in $\underline{\breve{\lambda}_{N,1}}$).
Thus we obtain the formula
\begin{equation}
\underline{\breve{\lambda}_{N,1}}\wedge\underline{\breve{\lambda}_{N,2}}=|z_1-z_N|^2|z_1-\overline{z}_N|^2
|z_2-z_N|^2|z_2-\overline{z}_N|^2\underline{\breve{\omega}_{N,1}}\wedge\underline{\breve{\omega}_{N,2}}.
\nonumber
\end{equation}
Proceeding in this way we find
\begin{equation}\label{wedgeN1NN-1}
\underline{\breve{\lambda}_{N,1}}\wedge\underline{\breve{\lambda}_{N,2}}\wedge\cdots\wedge\underline{\breve{\lambda}_{N,N-1}}
=\prod\limits_{j=1}^{N-1}|z_j-z_N|^2|z_j-\bar{z}_N|^2\underline{\breve{w}_{N,1}}\wedge\underline{\breve{w}_{N,2}}
\wedge\cdots\wedge\underline{\breve{w}_{N,N-1}}.
\end{equation}
In addition, we have
\[
\breve{\lambda}_{N,N}=\breve{w}_{N,N}\breve{z}_N-\breve{z}_N\breve{w}_{N,N}+\sum\limits_{k=1}^{N-1}\breve{w}_{N,k}\breve{t}_{N,k}
+d\breve{z}_N.
\]
The sum in the expression above does not contribute to the wedge product
$\underline{\breve{\lambda}_{N,1}}\wedge\underline{\breve{\lambda}_{N,2}}\wedge\cdots\wedge\underline{\breve{\lambda}_{N,N}}$,
since for each $k$ $(1\leq k\leq N-1)$ the one-forms consisting of $\breve{w}_{N,k}$ have already appeared
in the wedge  product  $\underline{\breve{\lambda}_{N,1}}\wedge\underline{\breve{\lambda}_{N,2}}\wedge\ldots\wedge\underline{\breve{\lambda}_{N,N-1}}$,
see equation (\ref{wedgeN1NN-1}).
Moreover, we note that
\[
\breve{w}_{N,N}\breve{z}_N-\breve{z}_N\breve{w}_{N,N}+d\breve{z}_N=
\begin{pmatrix}
                                                                         dz_N & (z_N-\bar{z}_N)\overline w_{N,N}^{(2)} \\
                                                                         (z_N-\bar{z}_N)w_{N,N}^{(2)} & d\bar{z}_N
\end{pmatrix},
\]
so
\[
\underline{\breve{w}_{N,N}\breve{z}_N-\breve{z}_N\breve{w}_{N,N}+d\breve{z}_N}=|z_N-\bar{z}_N|^2dz_N\wedge d\overline{z}_N\wedge
w_{N,N}^{(2)}\wedge \overline w_{N,N}^{(2)}.
\]
This gives
\[
\bigwedge_{j=1}^N\underline{\breve{\lambda}_{N,j}}
=|z_N-\bar{z}_N|^2\prod\limits_{j=1}^{N-1}|z_j-z_N|^2|z_j-\bar{z}_N|^2
dz_N\wedge d\bar{z}_N
\wedge w_{N,N}^{(2)}\wedge\overline w_{N,N}^{(2)}
\bigwedge_{j=1}^{N-1}\underline{\breve{w}_{N,j}}.
\]
Replace $N$ by $N-1$ in the equation just written above. Then replace $N-1$ by $N-2$,  $N-2$ by $N-3$ and so on, and make the wedge products
of the resulting expressions. In this way, one obtains
\begin{equation}\label{lambdaijform}
\underset{i\geq j}{\bigwedge}\underline{\breve{\lambda}_{i,j}}
=\prod\limits_{j<i}|z_j-z_i|^2|z_j-\bar{z}_i|^2\prod\limits_{j=1}^N|z_j-\overline{z}_j|^2
\underline{dZ}\wedge \underline\Omega.
\end{equation}
We need to extend the formula for $\bigwedge_{i\geq j}\underline{\breve{\lambda}_{i,j}}$ to that for the whole wedge product
$\underline\Lambda=\bigwedge_{i,j}\underline{\breve{\lambda}_{i,j}}$. For this purpose we recall the structure of $\Lambda$ given by~\eqref{LambdaDef}, since all independent elements of $\Omega$ already have appeared in~\eqref{lambdaijform}, the extension is trivial and we obtain the stated theorem.
\end{proof}

\subsection{Products of  matrices and change of measure}

Let $M_1, \ldots, M_n$ be $n$ matrices taken from $\Mat\left(\Hq', N\times N\right)$. Recall that these matrices can be represented as
\[
M_a=\big[(\breve{M}_a)_{i,j}\big]_{i,j=1}^N,
\qquad
(\breve{M}_a)_{i,j}=
\begin{pmatrix}
(M_a)_{i,j}^{(1)} & -\overline{(M_a)}_{i,j}^{(2)} \\
(M_a)_{i,j}^{(2)} & \overline{(M_a)}_{i,j}^{(1)}
\end{pmatrix}
\]
Here $(M_a)_{i,j}^{(1)}$ and $(M_a)_{i,j}^{(2)}$ are complex numbers. Denote by $P$ the product,
\[
P=M_1M_2\cdots M_n.
\]
The integration of a function of $P$ with respect to the Lebesgue measure is the same as integrating against the following $4nN^2$
differential form
\[
\bigwedge_a\underline{dM_a}=\bigwedge_a\bigwedge_{i,j}\underline{(d\breve M_a)_{i,j}}.
\]
We apply the generalised Schur decomposition for the quaternion matrices $M_1, \ldots, M_n$ (see Corollary \ref{CorrollaryGeneralizedSchurDecompositionQuaternionMatrices}),
and obtain
\begin{equation*}
M_a=U_a(Z_a+T_a)U_{a+1}^*,
\end{equation*}
where for each $a\ (1\leq a\leq n)$ the matrix $U_a$ is a unitary symplectic matrix, $U_{n+1}=U_1$ and the matrices $Z_a$ and $T_a$ are defined by equations
(\ref{Zquaternion}) and (\ref{Tquaternion}) correspondingly.
We have
\[
dM_a=U_a(\Omega_aS_a-S_a\Omega_{a+1}+dS_a)U_{a+1}^*.
\]
In the formula above $\Omega_a$ and $S_a$ are defined in analogy to the single matrix case. In particular, we distinguish between the diagonal and strictly upper triangular part of $S_a$, see Remark~\ref{ZplusT}. Note also that $\Omega_{n+1}=\Omega_1$. Furthermore, we will use the abbreviations
\[
\Lambda_a\equiv\Omega_aS_a-S_a\Omega_{a+1}+dS_a \qquad\text{for}\qquad a=1,\ldots,n
\]
and
\begin{equation}\label{zshort}
(z)_i\equiv \prod_{a=1}^n(z_a)_i \qquad\text{for}\qquad i=1,\ldots,N.
\end{equation}

\begin{thm}
\label{TheoremProductQuaternionJacobiDeterminantalFormula}
The Jacobian determinantal formula for the product of $n$ quaternion Ginibre matrices is
\begin{equation}\label{ProductQuaternionJacobiDeterminantalFormula}
\bigwedge_a\underline{dM_a}
=\prod_{j<i}|(z)_j-(z)_i|^2 |(z)_j-\overline{(z)}_i|^2 \prod_i|(z)_i-\overline{(z)}_i|^2
\bigwedge_a\big(\,\underline{dZ_a} \wedge \underline{dT_a} \wedge \underline{\Omega_a}\,\big).
\end{equation}
\end{thm}

\begin{rem}
1) Formula (\ref{ProductQuaternionJacobiDeterminantalFormula}) is a generalization of formula (\ref{JacobianDeterminantalFormulaQuaternion1})
to the case of a product of $n$ matrices from $\Mat\left(\Hq',N\times N\right)$.\\
2) Since we omit combinatorial constants before differential forms, we do not need to fix the order
in the wedge products (change in order results in the change of the overall sign).
\end{rem}

\begin{proof}
Proposition \ref{PropositionTQ} implies
$
\bigwedge_a\underline{dM_a}
=\bigwedge_a \underline{\Lambda_a}
$,
thus it is enough  to find the Jacobian determinant for the change of variables from $\Lambda_a$ to $\Omega_a$ and $dS_a$.
Explicit formulae for the matrix entries $(\breve{\lambda}_a)_{i,j}$ can be written as follows
\begin{equation}\label{EQ1}
(\breve{\lambda}_a)_{i,j}=\sum_{k=1}^{j}(\breve{w}_a)_{i,k}(\breve{s}_a)_{k,j}-\sum_{k=i}^N(\breve{s}_a)_{i,k}(\breve{w}_{a+1})_{k,j}+(d\breve{s}_a)_{i,j},
\qquad a=1,\ldots,n,
\end{equation}
where $(d\breve{s}_a)_{i,j}=\breve 0$ for $i>j$.
Recall that $(\breve{w}_{n+1})_{k,j}=(\breve{w}_1)_{k,j}$.
In particular, formula \eqref{EQ1} gives
\begin{equation}\label{LAMDAN1}
(\breve{\lambda}_a)_{N,1}=(\breve{w}_a)_{N,1}(\breve{s}_a)_{1,1}-(\breve{s}_a)_{N,N}(\breve{w}_{a+1})_{N,1},\\
\end{equation}
For all $a=1,\ldots, n$ the matrices $(\breve{s}_a)_{1,1}$ and $(\breve{s}_a)_{N,N}$ are diagonal, namely
\[
(\breve{s}_a)_{1,1}=
\begin{pmatrix}
                           (z_a)_1 & 0 \\
                           0 & \overline{(z_a)}_1
\end{pmatrix}
\qquad\text{and}\qquad
(\breve{s}_a)_{N,N}=
\begin{pmatrix}
                           (z_a)_N & 0 \\
                           0 & \overline{(z_a)}_N
\end{pmatrix}.
\]
Taking this into account we see that equations (\ref{LAMDAN1}) can be explicitly rewritten as
\[
(\lambda_a)_{N,1}^{(1)}=(z_a)_1(w_a)_{N,1}^{(1)}-(z_a)_N(w_{a+1})_{N,1}^{(1)}
\quad\text{and}\quad
(\lambda_a)_{N,1}^{(2)}=(z_a)_1(w_a)_{N,1}^{(2)}-\overline{(z_a)_N}(w_{a+1})_{N,1}^{(2)}.
\]
In matrix form, we have
\[
\begin{pmatrix}
        (\lambda_1)_{N,1}^{(1)} \\
        (\lambda_2)_{N,1}^{(1)} \\
        \vdots \\
        (\lambda_n)_{N,1}^{(1)}
\end{pmatrix}
=\begin{pmatrix}
                (z_1)_1 & -(z_1)_N & 0 & \ldots & 0 \\
                0 & (z_2)_1 & -(z_2)_N & \ldots & 0 \\
                \vdots & \vdots & \ddots & \ldots & -(z_{n-1})_N \\
                -(z_n)_N & 0 & 0 & \ldots & (z_n)_1
\end{pmatrix}
\begin{pmatrix}
        (w_1)_{N,1}^{(1)} \\
        (w_2)_{N,1}^{(1)} \\
        \vdots \\
        (w_n)_{N,1}^{(1)}
\end{pmatrix}
\]
and
\[
\begin{pmatrix}
        (\lambda_1)_{N,1}^{(2)} \\
        (\lambda_2)_{N,1}^{(2)} \\
        \vdots \\
        (\lambda_n)_{N,1}^{(2)}
\end{pmatrix}
=\begin{pmatrix}
                (z_1)_1 & -\overline{(z_1)}_N & 0 & \ldots & 0 \\
                0 & (z_2)_1 & -\overline{(z_2)}_N & \ldots & 0 \\
                \vdots & \vdots & \ddots & \ldots & -\overline{(z_{n-1})}_N \\
                -\overline{(z_n)}_N & 0 & 0 & \ldots & (z_n)_1
\end{pmatrix}
\begin{pmatrix}
        (w_1)_{N,1}^{(2)} \\
        (w_2)_{N,1}^{(2)} \\
        \vdots \\
        (w_n)_{N,1}^{(2)}
\end{pmatrix}
\]
Using these formulae, and the fact that
\[
\left|\det\left(\begin{array}{ccccc}
            \alpha_1 & -\beta_1 & 0 & \ldots & 0 \\
            0 & \alpha_2 & -\beta_2 & \ldots & 0 \\
            \vdots & \vdots & \ddots & \ldots & -\beta_{n-1} \\
            -\beta_n & 0 & 0 & \ldots & \alpha_n
          \end{array}
\right)\right|^2=\left|\alpha_1\alpha_2\ldots\alpha_n-\beta_1\beta_2\ldots\beta_n\right|^2,
\]
we obtain
\begin{equation}\label{ProductN1}
\bigwedge_a\underline{(\breve{\lambda}_a)_{N,1}}=
|(z)_1-(z)_N|^2 |(z)_1-\overline{(z)}_N|^2
\bigwedge_a\underline{(\breve{w}_a)_{N,1}},
\end{equation}
where we use the notation introduced in~\eqref{zshort}.

Next step is to consider $i=N$ and $j=2$. Here, we have
\begin{equation}\label{LambdaN2}
(\breve{\lambda}_a)_{N,2}=(\breve{w}_a)_{N,2}(\breve{s}_a)_{2,2}-(\breve{s}_a)_{N,N}(\breve{w}_{a+1})_{N,2}+(\breve{w}_a)_{N,1}(\breve{s}_a)_{1,2}.
\end{equation}
The last term on the right-hand side of equation \eqref{LambdaN2} does not contribute to the wedge product
$
\bigwedge_a\underline{(\breve{\lambda}_a)_{N,1}}\bigwedge_a\underline{(\breve{\lambda}_a)_{N,2}},
$
since the one-forms consisting of $(\breve{w}_a)_{N,1}$ already have appeared in $\bigwedge_a\underline{(\breve{\lambda}_a)_{N,1}}$.
It follows that~\eqref{LambdaN2} contributes to the total wedge product with a term analogue to~\eqref{ProductN1}. 
Repeating these considerations we obtain
\begin{equation}\label{LNN0}
\bigwedge_{j=1}^{N-1}\bigwedge_a\underline{(\breve{\lambda}_a)_{N,j}}
=\prod_{j=1}^{N-1}|(z)_j-(z)_N |^2  |(z)_j-\overline{(z)}_N|^2
\bigwedge\limits_{j=1}^{N-1}\bigwedge_a\underline{(\breve{w}_a)_{N,j}}.
\end{equation}
Now, let us compute the wedge product 
$
 \bigwedge_{j=1}^{N}\bigwedge_a\underline{(\breve{\lambda}_a)_{N,j}}
$. 
The formula for $(\breve{\lambda}_a)_{N,N}$ is
\begin{equation}\label{LNN}
(\breve{\lambda}_a)_{N,N}=(\breve{w}_a)_{N,N}(\breve{s}_a)_{N,N}-(\breve{s}_a)_{N,N}(\breve{w}_{a+1})_{N,N}+\sum\limits_{k=1}^{N-1}(\breve{w}_a)_{N,k}(\breve{s}_a)_{k,N}+(d\breve{s}_a)_{N,N}.
\end{equation}
The sum on the right-hand side of equation (\ref{LNN}) does not contribute to the wedge product, since all $(\breve{w}_a)_{N,k}$ already appear in~\eqref{LNN0}. Thus we only need to find the contribution of the first three terms in the right-hand side of equation (\ref{LNN}) to the wedge product. The procedure is the same as for a single matrix and we find
\begin{multline}\label{lambdaNN}
\bigwedge_a\underline{(\breve{w}_a)_{N,N}(\breve{z}_a)_N-(\breve{z}_a)_N(\breve{w}_{a+1})_{N,N}+(d\breve{z}_a)_N}\\
=|(z)_N-\overline{(z)}_N|^2
\bigwedge_a\Big((dz_a)_N\wedge \overline{(dz_a)_N}\wedge (w_a)_{N,N}^{(2)}\wedge \overline{(w_a)_{N,N}^{(2)}}\Big).
\end{multline}
Combining~\eqref{LNN0} and~\eqref{lambdaNN}, we obtain
\begin{multline}\label{z1}
\bigwedge_{j=1}^{N}\bigwedge_a\underline{(\breve{\lambda}_a)_{N,j}}=
|(z)_N-\overline{(z)}_N|^2
\prod_{j=1}^{N-1}|(z)_j-(z)_N|^2  |(z)_j-\overline{(z)}_N|^2
\\  \times
\bigwedge_a\Big((dz_a)_N\wedge \overline{(dz_a)}_N\wedge (w_a)_{N,N}^{(2)}\wedge \overline{(w_a)}_{N,N}^{(2)}\Big)
\bigwedge\limits_{j=1}^{N-1}\bigwedge_a\underline{(\breve{w}_a)_{N,j}}.
\end{multline}
Replacing $N$ in equation (\ref{z1})  first by $N-1$, then by $N-2$ and so on we see that equation (\ref{z1}) can be extended to
\begin{equation}\label{z2}
\bigwedge_{i\geq j}\bigwedge_a\underline{(\breve{\lambda}_a)_{i,j}}=
\prod_{i=1}^N|(z)_i-\overline{(z)}_i|^2
\prod_{j<i}|(z)_j-(z)_i|^2  |(z)_j-\overline{(z)}_i|^2 \bigwedge_a\big(\,\underline{dZ_a}\wedge \underline{\Omega_a}\,\big).
\end{equation}
It remains to compute the wedge product between $\bigwedge_{i\geq j}(\bigwedge_a\underline{\breve{\Lambda}_{i,j}(a)})$
and $\bigwedge_{i< j}(\bigwedge_a\underline{\breve{\Lambda}_{i,j}(a)})$.
Similar to the single matrix case we note that all independent quaternion one-forms $(\breve w_a)_{i,j}$ ($a=1,\ldots,n$ and $1\leq j\leq i\leq N$) have already appeared in~\eqref{z2}. With this observation, the wedge product becomes trivial and the theorem follows. 
\end{proof}

\subsection{Proof of Theorem \ref{PropositionQuaternionJointDensityOfEigenvalues}}

Let $M_1,\ldots, M_n\in\Mat(\Hq', N\times N)$ be independent random matrices each taken from the induced quaternion Ginibre ensemble with parameters $m_1, \ldots, m_n$. Consider the matrix $P_n^{\quaternion}$ defined in Section \ref{SectionDefinitionOfPquaternion}. We use the generalised Schur decomposition for matrices from $\Mat\left(\Hq', N\times N\right)$, see Corollary \ref{CorrollaryGeneralizedSchurDecompositionQuaternionMatrices}. Thus we assume that $M_1,\ldots, M_n$ are written as in the statement of Corollary \ref{CorrollaryGeneralizedSchurDecompositionQuaternionMatrices}, and obtain
\[
\Tr M_a^*M_a=\Tr Z^*_aZ_a+\Tr T^*_aT_a,
\]
where $Z_a$ and $T_a$ are defined by equations \eqref{Zquaternion} and \eqref{Tquaternion}. Taking this into account, and using the Jacobian determinantal
formula for the product of quaternion matrices in Theorem \ref{TheoremProductQuaternionJacobiDeterminantalFormula} we obtain that the density of
$P_n^{\quaternion}$ is proportional to
\begin{multline*}
\e^{-\frac{1}{2}\sum_{a=1}^n\Tr Z^*_aZ_a}
\prod_{i=1}^N\prod_{a=1}^n|(z_a)_i|^{2m_a}   
\prod_{i=1}^N|(z)_i-\overline{(z)}_i|^2 \prod_{j<i}|(z)_j-(z)_i|^2  |(z)_j-\overline{(z)}_i|^2
\bigwedge_a\underline{dZ_a} \\
\times \e^{-\frac{1}{2}\sum_{a=1}^n\Tr T^*_aT_a}
\bigwedge_a\big(\,\underline{dT_a}\wedge \underline{\Omega_a}\,\big).
\end{multline*}
It follows that the statistical properties of the eigenvalues of the product matrix do not depend on the unitary matrices $U_a\ (a=1,\ldots,n)$ nor on the strictly upper triangular matrices $T_a\ (a_1,\ldots,n)$. For this reason, the integration over these variables contributes only to the normalization constant.
Now the same arguments as in Section \ref{SectionProofTheoremJointDensity} lead us to formula (\ref{JointDensityOfEigenvaluesQ})
for the joint density of eigenvalues of $P_n^{\quaternion}$. Theorem \ref{PropositionQuaternionJointDensityOfEigenvalues} is proved.
\qed

\section{The correlation kernel for the generalised quaternion Ginibre ensemble}

In this  Section we prove Theorem \ref{TheoremMatrixValued Kernel}. Let $f(\zeta)$ be some finitely supported function defined on $\C$. Assume that $\zeta_1, \ldots, \zeta_N$
are complex random variables with the joint density given by equation (\ref{JointDensityOfEigenvaluesQ}).
Starting from equation (\ref{JointDensityOfEigenvaluesQ}), and using standard methods of Random Matrix Theory
(see Tracy and Widom \cite{tracy}) we find
\begin{equation}
\mathbb{E}\bigg(\prod\limits_{j=1}^N\left(1+f(\zeta_j)\right)\bigg)
=\frac{\Pf\left(\int(\zeta^j\bar\zeta^k-\zeta^k\bar\zeta^j)(\zeta-\bar\zeta)
(1+f(\zeta))w_n^{(m_1,\ldots,m_n)}(\zeta)d^2\zeta\right)_{0\leq j,k\leq 2N-1}}{\Pf\left(\int(\zeta^j\bar\zeta^k-\zeta^k\bar\zeta^j)(\zeta-\bar\zeta)
w_n^{(m_1,\ldots,m_n)}(\zeta)d^2\zeta\right)_{0\leq j,k\leq 2N-1}}.
\nonumber
\end{equation}
Introduce the matrix $Q=\left(Q_{j,k}\right)_{0\leq j,k\leq 2N-1}$ by
\[
Q_{j,k}=\int(\zeta^j\bar\zeta^k-\zeta^k\bar\zeta^j)(\zeta-\bar\zeta)
w_n^{(m_1,\ldots,m_n)}(\zeta)d^2\zeta.
\]
By the same argument as in Tracy and Widom \cite{tracy}, \S 8, we find
\begin{equation}
\mathbb{K}_N^{(m_1,\ldots,m_n)}(z,\zeta)=(\zeta-\bar\zeta)w_N^{(m_1,\ldots,m_n)}(\zeta)
\left(\begin{array}{cc}
        \sum\limits_{j,k=0}^{2N-1}\bar z^j\left(Q^{-1}\right)_{j,k}\zeta^k & -\sum\limits_{j,k=0}^{2N-1}\bar z^j\left(Q^{-1}\right)_{j,k}\bar\zeta^k \\
        \sum\limits_{j,k=0}^{2N-1}z^j\left(Q^{-1}\right)_{j,k}\zeta^k & -\sum\limits_{j,k=0}^{2N-1} z^j\left(Q^{-1}\right)_{j,k}\bar\zeta^k
      \end{array}
\right).
\nonumber
\end{equation}
The matrix elements $Q_{j,k}$ can be written as
\[
Q_{j,k}=2h_{j+1}^{(m_1,\ldots,m_n)}\delta_{k,j+1}-2h_k^{(m_1,\ldots,m_n)}\delta_{k+1,j},
\]
where $h_k^{(m_1,\ldots,m_n)}$ are defined by equation (\ref{Hk}). From this representation  it is not hard to see that the matrix $Q^{-1}$ is an antisymmetric matrix defined by the following conditions
\begin{itemize}
  \item $\left(Q^{-1}\right)_{2i,2j+1}=\displaystyle{\frac{2^{n(j-i)}}{2\pi^n}\prod\limits_{a=1}^n
  \frac{\Gamma\left(\frac{m_a}{2}+j+1\right)}{\Gamma\left(m_a+2j+2\right)\Gamma\left(\frac{m_a}{2}+i+1\right)},\; 0\leq i\leq j\leq N-1}$,
  \item $\left(Q^{-1}\right)_{2j+1,2i}=-\left(Q^{-1}\right)_{2i,2j+1},\; 0\leq i\leq j\leq N-1$,
  \item All other matrix elements are equal to zero.
\end{itemize}
Using the explicit formulae just written above for the matrix elements of $Q^{-1}$ we get the formula for the correlation kernel in the statement of Theorem \ref{TheoremMatrixValued Kernel}. \qed

\section{Hole and overcrowding estimates for the generalised quaternion ensemble}
\label{sec:holes-beta4}

In this Section we complete the task of extending the results obtained for products of complex matrices to the case of products of matrices from the induced quaternion Ginibre ensemble. We begin from the following due to Rider~\cite{Rider}:
\begin{prop}\label{PropositionIntegralI}
Set
$
z_i=r_i\e^{i\theta}\ (i=1,\ldots, N),
$
then we have
\begin{equation}\label{IntegralI}
\int_0^{2\pi}\ldots\int_0^{2\pi}\prod\limits_{k=1}^N(\overline{z}_k-z_k)
V(z_1,\overline{z}_1,\ldots,z_N,\overline{z}_N)d\theta_1\ldots d\theta_N=(4\pi)^N\per\left[r_i^{4j-2}\right]_{i,j=1}^N,
\end{equation}
where
\begin{equation}\label{Vandermond}
V(x_1,\ldots,x_N)=\prod\limits_{1\leq i<j\leq N}(x_i-x_j)=(-1)^{\frac{N(N-1)}{2}}\det\left(x_j^{i-1}\right)_{i,j=1}^N.
\end{equation}
\end{prop}
\begin{proof}
Denote by $I_N(r_1,\ldots,r_N)$ the integral in the left-hand side of equation \eqref{IntegralI}.
We have
\begin{align*}
I_N(r_1,\ldots,r_N)&=\int_0^{2\pi}\cdots\int_0^{2\pi}\prod\limits_{k=1}^N(\overline{z}_k-z_k)
V(z_1,\overline{z}_1,\ldots,z_N,\overline{z}_N)d\theta_1\cdots d\theta_N \\
&=(-1)^{N}\int_0^{2\pi}\cdots\int_0^{2\pi}\prod_{k=1}^N r_k(\e^{-i\theta_k}-\e^{i\theta_k}) \\
&\qquad\qquad\times\sum_{\sigma\in S(2N)}\sgn(\sigma) \prod_{n=1}^N\e^{i(\sigma(2n-1)-\sigma(2n))\theta_n}r_n^{\sigma(2n-1)+\sigma(2n)-2}d\theta_1\cdots d\theta_N.
\end{align*}
Integrating over variables $\theta_1, \ldots, \theta_N$ we obtain
\begin{equation}\label{I_N1}
I_N(r_1,\ldots,r_N)=(2\pi)^N\sum\limits_{\sigma\in S(2N)}\sgn(\sigma)
\prod_{n=1}^N(\delta_{\sigma(2n-1)+1,\sigma(2n)}-\delta_{\sigma(2n-1)-1,\sigma(2n)})r_n^{\sigma(2n-1)+\sigma(2n)-2}.
\end{equation}
Now formula (\ref{IntegralI}) follows by simple
combinatorial arguments.
We expand the product of Kronecker deltas
inside the sum in equation (\ref{I_N1}), and represent this expression as a sum of $2^N$ terms.
One of these terms is
\[
\delta_{\sigma(1)+1,\sigma(2)}
\delta_{\sigma(3)+1,\sigma(4)}
\cdots
\delta_{\sigma(2N-1)+1,\sigma(2N)}.
\]
It is not hard to see that the contribution of this term to $I_N(r_1,\ldots,r_N)$ is
\[
(2\pi)^N\sum\limits_{\sigma\in S(N)}r_{\sigma(1)}^{1+2-1}r_{\sigma(2)}^{3+4-1}\cdots r_{\sigma(N)}^{2N-1+2N-1}=(2\pi)^N\per\left[r_i^{4j-2}\right]_{i,j=1}^N.
\]
Then we check that each of these $2^N$ terms gives the same contribution to $I_N(r_1,\ldots,r_N)$. The formula~\eqref{IntegralI} follows.
\end{proof}

Now, we can turn to the proof of Theorem \ref{TheoremQIndependence}.
We note that
\[
V(z_1,\overline{z}_1,\ldots,z_N,\overline{z}_N)=\prod\limits_{i=1}^N(z_i-\overline{z}_i)
\prod\limits_{1\leq i<j\leq N}|z_i-z_j|^2|z_i-\overline{z}_j|^2,
\]
and that the expression for the joint probability density function (equation (\ref{JointDensityQuaternion})) can be rewritten as
\begin{equation}\label{JointDensityQuaternion1}
\varrho_{N,\quaternion}(z_1,\ldots,z_N)=\frac{1}{Z_N^{\quaternion}}\prod\limits_{k=1}^Nw(|z_k|)(\overline{z}_k-z_k)
V(z_1,\overline{z}_1,\ldots,z_N,\overline{z}_N).
\end{equation}
The formula in the statement of Proposition \ref{PropositionIntegralI} implies that the joint density of the unordered moduli $(r_i)_{i=1,\ldots,N}$ can be written as
\[
\frac{(4\pi)^N}{Z_N}\per[r_i^{4j-1}]_{i,j=1}^N\prod\limits_{j=1}^Nw(r_j).
\]
From this expression for the joint density the formula for $Z_N$ follows immediately. Moreover,
this expression implies that the joint density can be written as
\[
\frac{1}{N!}\per\left[\frac{r_i^{4j-1}w(r_i)}{\int_0^{\infty}r^{4j-1}w(r)dr}\right]_{i,j=1}^N.
\]
In particular,
the vector of squares of  $(r_i)_{i=1,\ldots,N}$ has the density
\[
\frac{1}{N!}\per\left[q_{j,\quaternion}(y_i)\right]_{i,j=1}^N,
\]
where the functions $q_{j,\quaternion}(y)$, $1\leq j\leq N$, are probability density functions defined by
\[
q_{j,\quaternion}(y)=\left\{
               \begin{array}{ll}
                 \frac{y^{2j-1}w(\sqrt{y})}{h_{j-1}^{\quaternion}}, & y\geq 0 \\
                 0, & y<0.
               \end{array}
             \right.
\]
The rest of the proof of Theorem \ref{TheoremQIndependence} is the repetition of arguments from the proof of Theorem \ref{TheoremIndependence}.

Now to prove  Theorem \ref{TheoremQ|z_1||z2|Distribution}  we use Theorem \ref{TheoremQIndependence}, Proposition \ref{PropositionSpringerThompson}, and proceed in the same way as in the proof
of Theorem \ref{CorollaryExplicitDistribution}.

Comparing Theorem~\ref{CorollaryExplicitDistribution} for complex matrices with Theorem \ref{TheoremQ|z_1||z2|Distribution} for quaternion matrices, we see that they differ by a simple scaling $k\to2k$; this scaling immediately yields Theorem~\ref{PropositionQuaternionExactFormulae}. The asymptotic formulae for the hole probability and overcrowding for quaternions given by Theorem~\ref{TheoremDecayQFiniteN} to Theorem~\ref{TheoremOvercrowdingEstimatesQ} are also directly linked to their complex analogues. The proofs are straightforward generalizations of those for complex matrices given in Section~\ref{SectionProofsHole} and~\ref{SectionProofsOvercrowding}. In fact, the proofs follow exactly the same steps as for the complex matrices, and for this reason we will not write them out explicitly. Instead, we have added footnotes in Section~\ref{SectionProofsHole}, which notify when the proofs for quaternion matrices differs.


\noindent
{\bf Acknowledgements:} We would like to thank Alon Nishry for discussions. The
SFB$|$TR12 ``Symmetries and Universality
in Mesoscopic Systems'' of the German research council DFG \& 
DAAD International Network ``From Extreme Matter to Financial Markets''
(G.A.), ``Stochastics and Real World Models'' IRTG 1132 (J.R.I) are acknowledged for financial support.
The School of Mathematics at the IAS Princeton is thanked for its kind hospitality while part of this work was written up.


\appendix

\section{}
\label{App}

In this appendix we collect some known results about the Euler--Maclaurin summation formula as we need them throughout the main part.

\begin{thm}\label{TheoremApostol1}(Second-derivative form of the Euler--Maclaurin summation formula).\\
Assume that $f$ is a function with a continuous second derivative on the interval $[1,n]$. In addition, assume that the improper integral
\[
\int_1^{\infty}|f^{(2)}(x)|dx
\]
converges. Then we have
\[
\sum_{k=1}^nf(k)=\int_1^nf(x)dx+D(f)+E_f(n),
\]
where
\[
D(f)=\frac{1}{2}f(1)-\frac{1}{2}P_2(0)f^{'}(1)-\frac{1}{2}\int_1^{\infty}P_2(x)f^{(2)}(x)dx,
\]
and
\[
E_f(n)=\frac{1}{2}f(n)+\frac{1}{2}P_2(0)f^{'}(n)+\frac{1}{2}\int_n^{\infty}P_2(x)f^{(2)}(x)dx.
\]
Here $P_2(x)$ is the second Bernoulli periodic function defined by
\[
P_2(x)=B_2(x-\lfloor x\rfloor),
\]
with $B_2(x)=x^2-x+\frac{1}{6}$ and $\lfloor x\rfloor$ denotes the integer part of $x$..
\end{thm}
\begin{proof} See Apostol \cite{Apostol}, Section 4.
\end{proof}
\begin{thm}\label{TheoremApostol2}(General form  of the Euler--Maclaurin summation formula).\\
Assume that $f$ is a function with a continuous  derivative of order $2m+1$ on the interval $[1,n]$.
In addition, assume that the improper integral
\[
\int_1^{\infty}|f^{(2m+1)}(x)|dx
\]
converges. Then we have
\[
\sum_{k=1}^nf(k)=\int_1^nf(x)dx+D(f)+E_f(n),
\]
where
\[
D(f)=\frac{1}{2}f(1)-\sum_{r=1}^m\frac{B_{2r}}{(2r)!}f^{(2r-1)}(1)
+\frac{1}{(2m+1)!}\int_1^{\infty}P_{2m+1}(x)f^{(2m+1)}(x)dx,
\]
and
\[
E_f(n)=\frac{1}{2}f(n)+\sum\limits_{r=1}^m\frac{B_{2r}}{(2r)!}f^{(2r-1)}(n)
-\frac{1}{(2m+1)!}\int\limits_n^{\infty}P_{2m+1}(x)f^{(2m+1)}(x)dx.
\]
Here $P_k(x)$ stand for  the Bernoulli periodic functions (periodic extensions of the Bernoulli polynomials $B_k(x)$) defined by
\[
P_k(x)=B_k(x-\lfloor x\rfloor).
\]
The constants $B_k=P_k(0)=P_k(1)$ are the Bernoulli numbers.
\end{thm}
\begin{proof} See Apostol \cite{Apostol}, Section 5.
\end{proof}
\begin{prop} The following asymptotic formula holds true
\begin{equation}\label{AsymptoticsSum1}
\sum\limits_{k=1}^{n}\log(m+k)=n\log(n)-n+\left(m+\tfrac{1}{2}\right)\log(n)+C(m)+O\left(n^{-1}\right),
\end{equation}
as $n\rightarrow\infty$. Here $C(m)$ is a constant given by
\begin{equation}
C(m)=m+1-\frac{1}{12(m+1)}-\left(\tfrac{1}{2}+m\right)\log(1+m)
+\frac{1}{2}\int\limits_1^{\infty}\frac{P_2(x)dx}{(m+x)^2}.
\end{equation}
In addition, we have
\begin{equation}\label{AsymptoticsSum2}
\sum\limits_{k=1}^n(k+m)\log(k+m)=\frac{1}{2}n^2\log(n)-\frac{n^2}{4}+n\left(m+\tfrac{1}{2}\right)\log(n)\\
+\frac{m(m+1)}{2}\log(n)+O(1),
\end{equation}
as $n\rightarrow\infty$.
\end{prop}
\begin{proof}
To obtain the first formula, equation (\ref{AsymptoticsSum1}), use the second-derivative form of the Euler-Maclaurin summation formula, Theorem \ref{TheoremApostol1}.
Application of the general form  of the Euler-Maclaurin summation formula with $m=1$ (Theorem \ref{TheoremApostol2}) gives equation (\ref{AsymptoticsSum2}).
\end{proof}

\raggedright

\end{document}